\documentclass[a4paper,11pt]{elsarticle}

\usepackage{amsmath,amssymb,amsfonts}
\usepackage{graphicx}
\usepackage{textcomp}
\usepackage{xcolor}
\def\BibTeX{{\rm B\kern-.05em{\sc i\kern-.025em b}\kern-.08em
    T\kern-.1667em\lower.7ex\hbox{E}\kern-.125emX}}
    
\usepackage{bm}
\usepackage[misc]{ifsym}
\usepackage{multirow}
\usepackage{stmaryrd}
\usepackage{amsthm}
\usepackage{subfigure}
\usepackage{float}
\usepackage[linesnumbered,boxed,ruled,vlined]{algorithm2e}
\usepackage[noend]{algpseudocode}
\usepackage{enumitem}
\usepackage{booktabs}
\usepackage[inkscapelatex=false]{svg}
\usepackage{hyperref}
\usepackage{soul}
\usepackage{tikz}
\usepackage{array}
\usepackage{siunitx}
\usepackage{empheq}
\usepackage{setspace}
\usepackage[misc]{ifsym}
\usepackage{natbib}

\makeatletter
\algnewcommand{\LineComment}[1]{\Statex \hskip\ALG@thistlm \(\triangleright\) #1}
\makeatother

\newtheorem{lemma}{Lemma}

\newtheorem{definition}{Definition}

\begin{document}
\onehalfspacing

\begin{frontmatter}



\title{Exact Clique Number Manipulation via Edge Interdiction}

\author[aff1]{Yi Zhou\corref{cor1}}
\ead{zhou.yi@uestc.edu.cn}

\author[aff1]{Haoyu Jiang}
\ead{FirstSSAT@outlook.com}

\author[aff1]{Chenghao Zhu}
\ead{axs7384@gmail.com}

\author[aff2]{Andr\'e Rossi}
\ead{andre.rossi@dauphine.psl.eu}

\cortext[cor1]{Corresponding author.}

\affiliation[aff1]{organization={School of Computer Science and Engineering, University of Electronic Science and Technology of China},
            city={Chengdu},
            country={China}}
 \affiliation[aff2]{organization={LAMSADE UMR CNRS 7243 Paris Dauphine University – PSL}, 
             city={Paris},
             country={France}}

\begin{abstract}
The Edge Interdiction Clique Problem (EICP) aims to remove at most $k$ edges from a graph so as to minimize the size of the largest clique in the remaining graph. This problem captures a fundamental question in graph manipulation: which edges are structurally critical for preserving large cliques? Such a problem is also motivated by practical applications including protein function maintenance and image matching. The EICP is computationally challenging and belongs to a complexity class beyond NP. Existing approaches rely on general mixed-integer bilevel programming solvers or reformulate the problem into a single-level mixed integer linear program. However, they are still not scalable when the graph size and interdiction budget $k$ grow. To overcome this, we investigate new mixed integer linear formulations, which recast the problem into a sequence of parameterized Edge Blocker Clique Problems (EBCP). This perspective decomposes the original problem into simpler subproblems and enables tighter modeling of clique-related inequalities. Furthermore, we propose a two-stage exact algorithm, \textsc{RLCM}, which first applies problem-specific reduction techniques to shrink the graph and then solves the reduced problem using a tailored branch-and-cut framework. Extensive computational experiments on maximum clique benchmark graphs, large real-world sparse networks, and random graphs demonstrate that \textsc{RLCM} consistently outperforms existing approaches. 
\end{abstract}







\end{frontmatter}


\section{Introduction}

The \textit{Maximum Clique Problem} (MCP) is a textbook combinatorial optimization problem that has been extensively studied in both theory and practice. 
Given a graph $G$, a clique is a subset of vertices in which every pair of distinct vertices is adjacent.
The MCP seeks to determine the size of the largest clique in a given graph $G$, commonly referred to as the \textit{clique number} $\omega(G)$.
From a theoretical perspective, MCP is NP-hard, hard to approximate~\citep{hastad1996clique} and W[1]-hard~\citep{downey2012parameterized}.
Despite these hardness results, a wide range of exact and heuristic algorithms have been developed over the past decades, many of which perform well on practical instances \citep{wu2015review,chang2019efficient,san2023clisat,LI20171}.

Recently, increasing attention has been devoted to problems related to manipulating the clique number of a graph through edge modifications \citep{zhu2025reduction, furini2019maximum, zhong2024interdicting, furini2021branch, becker2017bilevel}.
A basic variant of these problems considers either increasing or decreasing the clique number by adding or deleting a prescribed number of edges.
Since the clique number of a graph can only be increased by adding edges and decreased by deleting edges, a natural question is which $k$ edges should be added to maximize the clique number, or which $k$ edges should be removed to minimize it.

The former problem can be reformulated as finding the maximum vertex set such that its induced graph differs from a complete graph by at most $k$ missing edges. This problem is known in the literature as the \textit{maximum $k$-defective clique problem (MDCP)}  \citep{pattillo_clique_2013,gschwind2018maximum}. 
To the best of our knowledge, MDCP has been extensively studied in recent years and can be solved efficiently on large real-world graphs. \citep{chen2021computing,gao2022exact,chang2023efficient,dai2024theoretically,luo2024faster}.
In contrast, the problem of selecting at most $k$ edges for deletion so as to minimize the clique number of the resulting graph is referred to as the \textit{Edge Interdiction Clique Problem (EICP)}~\citep{furini2021branch}. 
Hence, the MDCP and the EICP appear to be \textit{dual} optimization problems, as shown in the following table.
\begin{table}[h]
    \centering
    \renewcommand{\arraystretch}{1.5}
    \begin{tabular}{|c|c|c|}
        \hline
        Problem & Definition & Manipulation \\
        \hline
        MDCP & $\displaystyle \max_{F\subseteq V\times V, |F|\le k} \omega((V, E\cup F))$ & Edge Addition \\
        \hline
        EICP & $\displaystyle \min_{F\subseteq V\times V, |F|\le k} \omega((V, E\setminus F))$ & Edge Deletion \\
        \hline
    \end{tabular}
    \label{tab:placeholder}
\end{table}

However, from a modeling perspective, the EICP is fundamentally more complex.
In particular, the EICP can be formulated as a bilevel optimization problem~\citep{dempe2020bilevel} in which two adversarial optimization problems are nested: the upper level minimizes $\omega(G)$, while the lower level computes a maximum clique in the resulting graph. This stands in contrast to the MDCP, which can be formulated as a single-level optimization problem.
An illustrative example is provided in Figure~\ref{fig:example}.
Beyond its interest in graph theory, the EICP is also motivated by practical considerations arising from a range of real-world applications. We highlight two representative examples below.

\begin{figure}[t]
    \centering
    \includegraphics[width=0.45\linewidth]{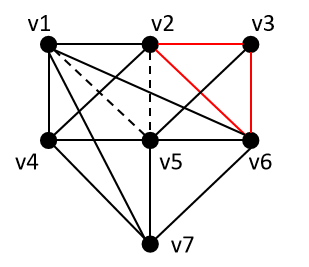}
    \caption{An example of EICP. In the original graph, the clique number $\omega(G)$ is equal to 4. If $k=2$, then an optimal strategy is to remove edges $(v_1, v_5)$ and $(v_2, v_5)$ (marked by dashed lines), which reduces $\omega(G)$ to 3 (highlighted by red edges).}
    \label{fig:example}
\end{figure}

\emph{Protein function maintenance.}
In protein structure analysis, cliques often correspond to highly conserved functional regions. Interdicting edges within a clique can be interpreted as disrupting specific molecular interactions,  providing a way to assess the resilience of protein functionality under partial interaction \citep{SAMUDRALA1998287,Anand2018}.
      
\emph{Image matching.} 
The image matching problem can be modeled as a clique-finding problem in an association graph, where maximum cliques represent consistent feature correspondences ~\citep{stentiford2014face,SanSegundo2015}. Solving the EICP allows one to identify critical matches and evaluate the robustness of the matching under edge removals. 

Although the EICP is important from multiple perspectives, the problem has received significantly less attention than MDCP in the literature. 
Perhaps the most effective exact and practically viable approach for solving the EICP was introduced by \citep{furini2021branch}, where a mixed-integer linear programming formulation combined with upper-bounding techniques was proposed. 
However, its performance on large real-world graphs has not been systematically evaluated, partly due to the lack of publicly available complete source code. Moreover, under practical time limits, the number of instances solvable by this approach decreases markedly as the interdiction budget $k$ increases within critical ranges.
An alternative approach was proposed in \citep{zhu2025reduction} that
transforms the EICP into the \textit{Node Interdiction Clique Problem} (NICP), defined as determining at most a given number of vertices such that their deletion minimizes the clique number of the remaining graph. 
This transformation enables the use of existing NICP solvers, such as those developed in \citep{furini2019maximum,MATTIA202448}. 
However, this transformation increases the number of edges from $|E|$ in the original graph to $\Theta(|E|^2)$ in the transformed graph, resulting in substantial computational overhead, particularly for large real-world graphs. 
As a consequence, this approach remains primarily of theoretical interest.
Another problem closely related to the EICP is the \textit{Edge Blocker Clique Problem} (EBCP) \citep{MahdaviPajouh2020}, which aims to minimize the number of edges removed so that the clique number of the remaining graph does not exceed a given threshold.
Clearly, the EICP can be solved by iteratively testing the EBCP under different thresholds. In this paper, we further develop this connection and show how it can be exploited algorithmically. Finally, as discussed earlier, the EICP belongs to the broader class of bilevel interdiction problems. In principle, general-purpose mixed-integer bilevel optimization solvers \citep{fischetti2017new,fischetti2018dynamic} can also be applied. We experimentally evaluate this option and conclude that it is not the most efficient approach for the EICP.

\subsection{Our Contributions}

Despite its relevance, efficient algorithms for the EICP remain limited. This paper addresses this gap with the following contributions.

\begin{itemize}
    \item 
    We reformulate the EICP as a sequence of parameterized Edge Blocker Clique Problems (EBCPs), where the objective is to remove a minimum number of edges so that the clique number falls below a given threshold $p$. We develop new MILP formulations for the EBCP, including a maximal-clique-based formulation, and show how its associated inequalities can be embedded as valid cuts within a branch-and-cut framework.

    \item 
    Based on this reformulation, we propose \textsc{RLCM}, a two-stage exact algorithm that combines graph reduction, novel lower-bound estimation, and an iterative branch-and-cut procedure with lazy constraint separation to efficiently solve the EICP.

    \item 
    Extensive experiments on DIMACS benchmark graphs, large real-world sparse networks, and random graphs demonstrate that \textsc{RLCM} consistently outperforms existing approaches. In particular, \textsc{RLCM} solves more hard DIMACS instances than competing algorithms for every tested value of $k$ within the 600-second time limit. On large real-world sparse networks, it achieves shorter running times than other solvers on 664 out of 695 instances. Ablation studies further confirm the effectiveness of the proposed components.

\end{itemize}

\section{Notations}

Consider a simple undirected graph $G = (V, E)$, where $V$ denotes the vertex set and $E$ the edge set. 
The \textit{neighborhood} of a vertex $v \in V$, denoted by $N_G(v)$, is the set of vertices adjacent to $v$. For notational simplicity, we use the abbreviated form $N(v)$ instead of $N_G(v)$ unless otherwise specified.
For any subset $S \subseteq V$ of vertices, the notation $G[S]$ refers to the subgraph of $G$ induced by $S$,  and the notation $E(S)$ refers to the edge set of $G[S]$.

Given a graph $G=(V,E)$, a set of vertices $C \subseteq V$ is called a \textit{clique} of $G$ if, for every pair of distinct vertices $u, v \in C$, the edge $\{u, v\}$ belongs to $E$. The \textit{clique number} of $G$, denoted by $\omega(G)$, is the cardinality of the largest clique in $G$.
As mentioned, the computation of $\omega(G)$ is well-studied. 
We use $\mathrm{MaxClique}(G)$ to denote a primitive that returns a clique of maximum size from the graph $G$.
In our implementation, $\mathrm{MaxClique}(G)$ is based on a portfolio of two well-performing open-source algorithms.
Specifically, when $|V| \leq 500$, we employ the \textit{MoMC} solver \citep{LI20171}, while for $|V|> 500$, we use the \textit{MC-BRB} solver \citep{chang2019efficient}. This portfolio strategy is designed to achieve robust performance across graphs of different sizes.

A clique is said to be \textit{maximal} if it is not a proper subset of any larger clique.
Throughout the paper, when the context is clear, we slightly abuse notation by using $C$ to denote both a clique and the complete subgraph it induces, \textit{i.e.}, $G[C]$. Given a graph $G=(V,E)$ and the nonnegative budget value $k\in \mathbb{N}$, $\eta(G,k)$ denotes the minimum clique number of the subgraph obtained by removing at most $k$ edges from $G$, \textit{i.e.}
\[
\eta(G, k) = \min_{\substack{F \subseteq E \\ |F| \leq k}} \omega(G'=(V,E\setminus F))
\]
The problem of computing $\eta(G, k)$, together with a corresponding set of edges $F^*$ attaining this minimum, is referred to as the \textit{Edge Interdiction Clique Problem} (EICP).


\section{The MILP formulations for EICP}

\subsection{An Existing Set-Covering Formulation of EICP}

In \citep{furini2021branch},  a set-covering MILP formulation for the EICP was proposed. 
The formulation assumes that valid lower and upper bounds on $\eta(G,k)$, denoted by  $lb$  and  $ub$, respectively, are available in advance. 
For each edge $e \in E$, a binary variable $x_e$ is introduced, where  $x_e=1$ if edge $e$ is interdicted and $x_e=0$ otherwise.  
In addition, for each $i \in \{lb, lb+1,...,ub\}$, the binary variable $y_i$ is defined such that $y_i=1$ if $\eta(G,k)=i$, and $y_i=0$ otherwise.  
By construction, exactly one variable $y_i$ must take the value 1, while all others are equal to 0.
Using these variables, the set-covering MILP formulation for the EICP, denoted by 
EICP-MILP, is given below.

\begin{align}
\textbf{(EICP-MILP)}\ \  \min \quad & \displaystyle\sum_{i=lb}^{ub} iy_i \nonumber \\
\text{s.t.} 
&\sum_{e \in E(C)} x_e \geq \sum_{i = lb }^{ub} \gamma_{clq}(|C|,i)y_{i}, \quad  \forall C \in \mathcal{C},\ |C|\ge 2, \label{base_constraint}  \\
& \sum_{e\in E} x_e \leq k,\label{cardi_constraint_1} \\
& \sum_{i=lb}^{ub} y_i = 1 \label{yi_constraint_1}\\
& y_i\in \{0,1\}, \quad \forall i\in \{lb,lb+1,\ldots, ub\}   \nonumber\\
&x_e \in\{0,1\}, \quad \forall e \in E \nonumber
\end{align}

In the formulation, $\mathcal{C}$ indicates the set of all cliques in $G$.
The function $\gamma_{clq}(n,i)$, where $n,i\in \mathbb{N}$ and $i\le n$, represents the minimum number of edges that must be removed from a complete graph of  $n$ vertices so that the clique number in the resulting graph is at most $i$. 
By Tur\'an's Theorem ~\citep{turan1941external},  $\gamma_{clq}(n,i)$ admits a closed-form expression given by  
 \[
\gamma_{clq}(n,i) =\begin{cases}
  n_{\alpha-1} \binom{\alpha-1}{2} + n_{\alpha} \binom{\alpha}{2}, & \text{if } i<n, \\
  0, & \text{if } i\ge n, \\
\end{cases}
 \]
 where $\alpha = \left\lceil \frac{n}{i} \right\rceil, n_{\alpha-1} = i\alpha - n,  n_{\alpha} = \frac{n - n_{\alpha-1} (\alpha-1)}{\alpha}$.

Hence, Inequality (\ref{base_constraint}) enforces that for every clique $C$ in $G$,  at least $\gamma_{clq}(|C|, i)$ edges in $E(C)$ must be interdicted whenever $\eta(G,k)=i$.
Inequalities (\ref{cardi_constraint_1}) and (\ref{yi_constraint_1})  bound the total number of interdicted edges by $k$ and ensure exactly one variable $y_i$ takes  value 1.
Regarding the bounds $lb$ and $ub$, Furini \textit{et al.} suggested a simple lower bound $lb=2$ for a non-trivial graph, and introduced a reformulated MILP to obtain a valid upper bound $ub$ \citep{furini2021branch}.

\subsection{A Parameterized Formulation of EICP} \label{subsec_parameter_form}

The EICP-MILP simultaneously encodes the resulting clique number through the variables 
$y_i$ and the interdicted edges through the variables $x_e$, which leads to a very large search space.
In this section, we decompose the problem by proposing a parameterized formulation that determines 
$\eta(G,k)$ in a sequential manner. The key idea is to reduce the EICP to the aforementioned \textit{Edge Blocker Clique Problem} (EBCP) \citep{MahdaviPajouh2020}, which is formally defined as follows.
\begin{definition}[Edge Blocker Clique Problem]\label{def:1}
Given an undirected graph $G=(V,E)$ and a parameter $p \in \mathbb{N}$, the Edge Blocker Clique Problem asks for a minimum-cardinality edge subset $F \subseteq E$ whose removal ensures that the clique number of the remaining graph $G'=(V,E \setminus F)$ is at most $p$, i.e. $\omega(G')\le p$.  Let us denote the minimum cardinality of the optimal edge set as $\gamma(G,p)$.
\end{definition}
In particular, when $G$ is a complete graph on $n$ vertices, we have $\gamma(G,p)=\gamma_{clq}(n,p)$. Hence, the EBCP admits a closed-form solution and can be solved in constant time for complete graphs.

Based on the EICP-MILP, a mixed-integer linear programming formulation for the EBCP can be readily derived as follows.
\begin{align}
\textbf{(EBCP-MILP)}\gamma(G,p)=\min \quad & \displaystyle\sum_{e \in E} x_e  \nonumber \\
\text{s.t.} 
& \sum_{e \in E(C)} x_e \geq \gamma_{clq}(|C|,p), \quad \forall C \in \mathcal{C}, \ |C| \geq 2, \label{EBCP2} \\
&x_e \in \{0,1\}, \quad \forall e \in E \nonumber
\end{align}
In this formulation, the binary variable $x_e$ indicates whether an edge $e\in E$ is blocked, and  $\mathcal{C}$ still denotes the collection of all cliques in the graph $G$. (We use the term \textit{block} to distinguish it from the term \textit{interdict} employed in the context of the EICP.) Inequality~\eqref{EBCP2} stipulates that blocking at least $\gamma_{clq}(|C|,p)$ edges  in each clique 
$C$ suffices to bound the clique number by $p$.

We further observe that, for a fixed graph $G$, $\gamma(G,p)$ is monotonically non-increasing in $p$. Consequently, if $\gamma(G,p)$ can be computed efficiently, then EICP can alternatively be solved by applying a binary search over $p$, requiring $O(\log(ub-lb))$ evaluations of $\gamma(G,p)$. Specifically, if $\gamma(G,p)\le k$, then $\eta(G,k)\le p$; otherwise, $\eta(G,k)>p$.

In EBCP-MILP, the number of Inequalities~(\ref{EBCP2}) is equal to the total number of cliques in the graph.
Since any subset of a clique is itself a clique (a property commonly referred to as \textit{hereditary}), a clique of size $n$ contains $2^n-n-1$ nontrivial subcliques.
If $\mathcal{C}_{\max}$ is the set of maximal cliques in $G$, then the total number of inequalities potentially required by the EBCP-MILP is $O(|\mathcal{C}_{\max}|2^{\omega{(G)}})$.

\subsection{A Reformulation of EBCP-MILP Using Maximal Cliques}
\label{subsection_ineq_max_clique}

We reformulate the EBCP-MILP so that only maximal cliques are required.
Let $C\in \mathcal{C}_{\max}$ be a maximal clique of $G$.
Given a permutation of the vertices of $C$, denoted by $\overrightarrow{C}$, we construct a directed graph by orienting each edge in the induced subgraph $G[C]$ according to this permutation, yielding a tournament.
Specifically, for every edge $e \in E(C)$, the edge is oriented from the vertex with lower rank vertex to the vertex with higher rank vertex in $\overrightarrow{C}$.
For a vertex $u\in C$, let $N_{\overrightarrow{C}}(u)$ denote its set of outgoing neighbors in the resulting directed graph.
Based on $\overrightarrow{C}$, we introduce binary variables $z_{u,\overrightarrow{C}}$ and Inequalities (\ref{sb:1})--(\ref{sb:4:variable}).
Intuitively, $z_{u,\overrightarrow{C}}=1$ indicates that none of the outgoing edges of $u$ is blocked.
These constraints ensure that the number of vertices whose outgoing edges remain unblocked is bounded by $p$, thereby enforcing an upper bound on the clique number after edge blocking.
Lemma \ref{lemma_extra_inequality_hold} shows that these inequalities are valid for EBCP.

\begin{align}   
        & z_{u,\overrightarrow{C}}+\sum_{v\in N_{\overrightarrow{C}}(u)} x_{u,v} \ge 1, 
            && \forall\, u\in C, \label{sb:1}\\
        & z_{u,\overrightarrow{C}}+x_{u,v} \le 1, 
            && \forall\, u\in C,\ \forall\, v\in  N_{\overrightarrow{C}}(u), \label{sb:2}\\
        & \sum_{u\in C} z_{u,\overrightarrow{C}} \le p, \label{sb:3}\\
        & z_{u,\overrightarrow{C}} \in\{0,1\}, 
            && \forall\ u\in C \label{sb:4:variable} 
\end{align}


\begin{figure}[htbp]
    \centering
    \subfigure[A clique $C$ from the original graph. ]
    {
        \includegraphics[width=0.4\textwidth]{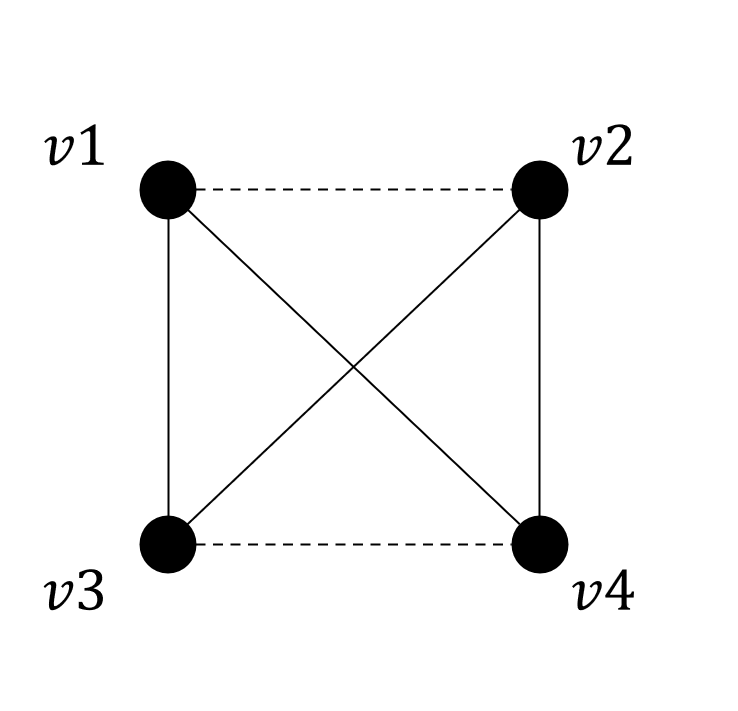}
        \label{fig:new_reformulate1}
    }
    \hfill
    \subfigure[The directed graphs $\overrightarrow{C}$ obtained from different permutations of $C$.]
    {
        \includegraphics[width=0.4\textwidth]{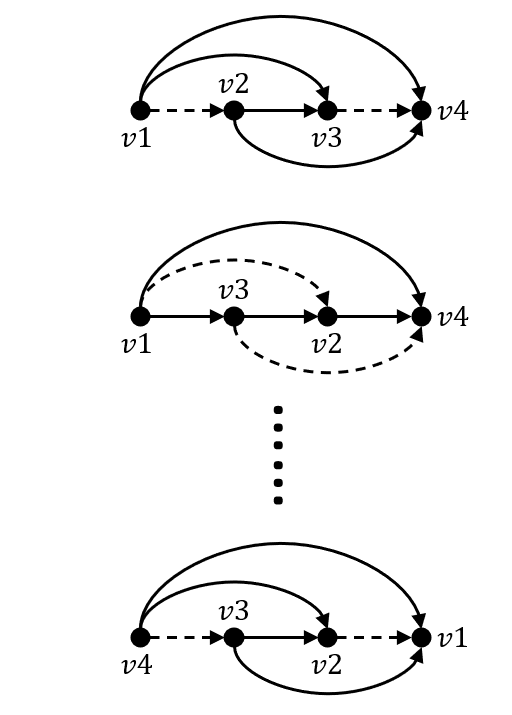}
        \label{fig:new_reformulate2}
    }
    \caption{An example of clique $C$ and the directed graphs.}
    \label{fig:new_reformulate}
\end{figure}

\begin{lemma}
\label{lemma_extra_inequality_hold}
    Let $C$ be a clique, and let $x_{e}=1$ indicate that an edge $e\in E(C)$ is blocked (removed), and $x_{e}=0$ otherwise. 
    Given an integer $p$, suppose that for every permutation $\overrightarrow{C}$ of $C$, there exists an assignment of variables $z$ such that Inequalities (\ref{sb:1})-(\ref{sb:3}) are satisfied. Then, there exists no subclique $C'\subseteq C$ with $|C'|\ge p$.
\end{lemma}

\begin{proof}    
    We prove this lemma by contradiction.
    Assume that after blocking edge $e$ where $x_e=1$, there is a clique in $C'\subseteq C$ such that $|C'|=p+1$. 
    Let $C'=\{v_1,...,v_{p+1}\}$. 
    Consider a permutation $\overrightarrow{C}$ of $C$ in which the vertices of $C'$ appear consecutively, in the order $v_1,...,v_{p+1}$. 
    Since  $C'$ is a clique in the remaining graph, none of the edges between vertices in $C'$ is blocked. Hence, $x_{v_i,v_j}=0, \forall 1\le i< j\le p+1$. 
    For this permutation $\overrightarrow{C'}$, it follows from Inequalities (\ref{sb:1}) and (\ref{sb:2}) that  $z_{u,\overrightarrow{C'}}=1, \forall u\in C'$. Therefore, $\sum_{u\in C} z_{u,\overrightarrow{C'}} \ge |C'|=p+1$, which violates Inequality (\ref{sb:3}). This contradiction completes the proof.
\end{proof}

\noindent\textit{Example.} Consider the clique $C$ illustrated in Figure \ref{fig:new_reformulate}(a). Suppose that $x_{v_1,v_2}=1$, $x_{v_3,v_4}=1$ and $x_e=0$ for other edges in $E(C)$; that is, the edges $\{v_1,v_2\}$ and $\{v_3,v_4\}$ are blocked under the current integer solution. Let $p=2$. 
Figure~\ref{fig:new_reformulate}(b) depicts some permutations of the vertices permutations of $C$. For each permutation of $\overrightarrow{C}$, we assign $z_{v,\overrightarrow{C}}=1$ if none of the outgoing edges of $v$ is blocked and $z_{v,\overrightarrow{C}}=0$ otherwise. It can be readily verified that, for every permutation $\overrightarrow{C}$, the inequality $\sum_{v\in C}z_{v,\overrightarrow{C}}\le 3$ is satisfied.

We are now ready to present a refined MILP formulation for the EBCP based solely on maximal cliques.
\begin{align}
\textbf{(EBCP-MILP-2)} \gamma(G,p)=\min \quad & \displaystyle\sum_{e \in E(G)} x_e  \nonumber \\
\text{s.t.} 
& \sum_{e \in E(C)} x_e \geq \gamma_{clq}(|C|,p), \quad \forall C \in \mathcal{C}_{\max}, \ |C| \geq 2  \\
& \text{Inequality (\ref{sb:1})-(\ref{sb:4:variable})}, \quad  \forall \text{permutation } \overrightarrow{C}, \forall C \in \mathcal{C}_{\max}, \ |C| \geq 2\\
&x_e \in \{0,1\}, \quad \forall e \in E \nonumber
\end{align}
By Lemma~\ref{lemma_extra_inequality_hold}, this formulation correctly enforces that the clique number of the graph obtained after edge blocking does not exceed 
$p$. Therefore, \textbf{EBCP-MILP-2} provides a valid formulation for EBCP.

\subsection{A Combination of Two Formulations in A Branch-and-bound Algorithm}
\label{subsec_combine_formulation}
Although EBCP-MILP-2 requires only maximal cliques, the total number of inequalities still remains bounded by $O(|\mathcal{C}_{\max}|\omega(G)!)$, which is asymptotically larger than that of the original EBCP-MILP.
To address this issue, we combine the two formulations within a branch-and-cut framework.
The general idea is specified as follows. We first extract a subset of maximal cliques form $G$ and build the formulation EBCP-MILP.
When an integer solution $\{x_e\}$ is obtained -- corresponding to a set of edges $F$ -- we construct the residual graph $G'=(V,E\setminus F)$ and compute a maximum clique $C$ from $G'$. 
If this clique violates Inequality (\ref{EBCP2}), we sample a small number of  permutations $\overrightarrow{C}$ (two in our implementation) and add both Inequality (\ref{EBCP2}) and  extra Inequalities (\ref{sb:1})-(\ref{sb:4:variable}) for each permutation to the MILP model.
Overall, this procedure introduces  $O(|C|)$ additional variable and $O(|C|^2)$ additional inequalities in a cutting-plane manner.
In the following section, we demonstrate that this strategy improves the efficiency of the algorithm, particularly on dense graphs.

\section{A Two-stage Algorithm - \textsc{RLCM}}
\label{section_two_stage_RLCM}

In this section, we describe the \textsc{RLCM} algorithm for solving the EICP. The algorithm follows a two-stage framework, with each stage addressing a distinct aspect of the problem.

\begin{itemize}
    \item \textit{Pre-processing}. Given the input graph $G$ and the budget value $k$, the first stage computes a lower bound $lb$ on $\eta(G,k)$ and applies a set of reduction rules based on this bound. These reductions are guaranteed to preserve optimality and therefore do not alter the optimal value of the original problem. 
    \item \textit{Iterative branch-and-cut}. In the second stage, \textsc{RLCM} estimates an upper bound $ub$ of $\eta(G,k)$ and then iteratively determines the smallest integer $p \in \{lb,\ldots,ub\}$ such that $\gamma(G,p) \le k$.  
    Specifically, the algorithm initializes $p = ub-1$ and solves the corresponding EBCP instance using a branch-and-cut procedure to compute $\gamma(G,p)$. If $\gamma(G,p) \le k$, the value of $p$ is decreased by one and the process continues. Otherwise, the algorithm terminates and returns $p+1$ as the optimal value of $\eta(G,k)$.
\end{itemize}
The two stages of \textsc{RLCM} are detailed in the following subsections, where we describe their components and integration in more depth.

\subsection{Lower Bound Estimation}

We propose a combinatorial procedure to compute a valid lower bound on $\eta(G,k)$, summarized in Algorithm~\ref{alg:get_disjoint_cliques}.
The method is based on the observation that, for any subgraph $G'$ of $G$, the value $\eta(G',k)$ provides a lower bound on $\eta(G,k)$, an idea previously exploited for vertex interdiction problems \citep{furini2019maximum}.

The algorithm constructs a subgraph $G'$ consisting of a collection of vertex-disjoint cliques $\mathcal{C}'\subseteq \mathcal{C}$.
For such a graph, the restricted EICP can be evaluated efficiently by exploiting the closed-form expression of $\gamma_{clq}(\cdot,\cdot)$ for complete graphs.
In particular, the minimum number of edges required to reduce the clique number of $G'$ to at most $p$ is given by
\[
\gamma(G',p)=\sum_{C\in \mathcal{C}'} \gamma_{clq}(|C|,p),
\]
which allows the smallest feasible value of $p$ to be determined via binary search. The overall time complexity of the procedure is $O(|V|^2+|\mathcal{C}'|\log(\omega(G)))=O(|V|^2)$ as $|\mathcal{C}'|<|V|$. Hence, the time is dominated by the construction of the disjoint cliques.

\begin{algorithm}[H]
    \caption{Lower bound on $\eta(G, k)$.}
    \label{alg:get_disjoint_cliques}
    \KwIn{Graph \( G = (V, E) \) and interdiction budget $k$}
    \KwOut{A lower bound value $lb$ on $\eta(G, k)$}
    \emph{EstimateLB}($G,k$) \\
    \Begin{
    Initialize an empty collection of set $ \mathcal{C'} \gets \{\emptyset\}$ \;
    \ForAll{\( v \in V \)}{
        \If{there is a $C\in \mathcal{C'}$ such that $C\cup \{v\}$ is a clique}{
            update $\mathcal{C'}$ by replacing $C$ with $C\cup\{v\}$             
        }\Else{
            add $\{v\}$ into $\mathcal{C'}$
        }
    }
    Let $\Delta=\max\{|C|:C\in \mathcal{C}'\}$ \;
    Use binary search to compute the smallest $p\in\{1,2,...,\Delta\}$ such that $\sum_{C\in \mathcal{C}'} \gamma_{clq}(|C|,p)$ is no more than $k$ \; 
    \Return{$p$}
    }
\end{algorithm}




\subsection{Graph Reduction}

After obtaining a valid lower bound $lb$ on $\eta(G,k)$, we apply graph reduction rules to simplify the input graph $G$.
The underlying principle is as follows: if a vertex or an edge of $G$ does not belong to any clique greater than $lb$, then it can be safely removed from $G$ without changing the value of $\eta(G,k)$. 
For notational purposes, let $\mathcal{F}(G,k)$ denote the family of optimal interdiction sets, \textit{i.e.}, $\mathcal{F}(G,k) =\{F^*\subseteq E|\omega((V, E\setminus F^*))=\eta(G,k)\}$.
We now present a vertex-based reduction rule.

\begin{lemma}[Vertex Clique Reduction] \label{Lemma:Vertex Clique Reduction}
Let  $G=(V,E)$ be a graph, and let  $k$ and $lb$ be two integers such that $lb$ is a valid lower bound on $\eta(G,k)$. 
If there exists a vertex $u\in V$ satisfying $ \omega(G[N(u)]) \leq lb-2 $, then $\eta(G,k)=\eta(G[V\setminus \{u\}],k)$.
\end{lemma}

\begin{proof}
Let \( G' = G[V \setminus \{u\}] \) and \(F\in \mathcal{F}(G', k) \). Since \( G' \) is a subgraph of \( G \), it follows that \( \eta(G, k) \geq \eta(G', k) \geq \eta(G,k) -1 \). 
Additionally, because \(F\in \mathcal{F}(G', k) \), we have \(\omega( (V \setminus \{u\} , E\setminus F)) = \eta(G', k) \). Moreover, because edge removals cannot increase clique sizes, \( \omega((N(u),E\setminus F))\leq\omega(G[N(u)]) \leq lb-2\). Since \(lb \le \eta(G,k)\) and \(\eta(G',k)\ge lb\), we obtain \( \omega((N(u),E\setminus F)) \leq\eta(G', k)-1\). 
Therefore, vertex $u$ cannot belong to any clique of size $\eta(G',k)$ in the graph $(V,E\setminus F)$, which implies \( \omega((V, E\setminus F)) = \omega((V \setminus \{u\}, E\setminus F)) = \eta(G', k) \). As a result, \( \eta(G, k) \leq \eta(G', k) \). Since \( \eta(G, k) \geq \eta(G', k) \) was already established, we conclude that \(\eta(G, k) = \eta(G', k)\).
\end{proof}

\begin{lemma}[Edge Clique Reduction] \label{Lemma:Edge Clique Reduction}
Let  $G=(V,E)$ be a graph, and let  $k$ and $lb$ be two integers such that $lb$ is a valid lower bound on $\eta(G,k)$.
If there exists an edge $\{ u, v \} \in E$ satisfying $ \omega(G[N(u)\cap N(v)]) \leq lb-3 $, then $\eta(G,k)=\eta((V,E \setminus \{ \{u, v\} \}),k)$.
\end{lemma}

\begin{proof}
Let \( G' = (V,E\setminus \{\{u,v\}\}) \) and \(F\in \mathcal{F}(G', k) \). Since \( G' \) is a subgraph of \( G \), it follows that \( \eta(G, k) \geq \eta(G’, k) \geq \eta(G,k) -1 \). 
Additionally, because \(F\in \mathcal{F}(G', k) \), we have \( \omega((V, E \setminus\{\{u, v\}\} \setminus F)) = \eta(G', k) \). We also have $\omega((V, E\setminus F) \geq \eta(G,k) \geq lb$. Due to $\omega((N(u)\cap N(v), E\setminus F))+2\leq\omega(G[N(u)\cap N(v)])+2\leq lb$, we have $\omega((V, E\setminus F))=\omega((V, E \setminus \{\{u, v \}\}\setminus F]) = \eta(G', k)$
As a result, \( \eta(G, k) \leq \omega((V, E\setminus F)) = \eta(G', k) \). 
Since \( \eta(G, k) \geq \eta(G', k) \) was already established, we conclude that \(\eta(G, k) = \eta(G', k)\).
\end{proof}

Based on the reduction rules established in Lemmas~\ref{Lemma:Vertex Clique Reduction} and~\ref{Lemma:Edge Clique Reduction}, we apply a preprocessing procedure summarized in Algorithm~\ref{alg:reduce}. 
The algorithm first computes a valid lower bound $lb$ and then iteratively removes vertices and edges whose local neighborhoods cannot support cliques larger than $lb$.
Both degree-based and coloring-based upper bounds are exploited to perform these reductions without explicitly computing clique numbers.
The procedure is widely used for computing upper bounds on the clique number of a graph \citep{li2010efficient,furini2019maximum,chang2019efficient,chen2021computing}
For the remaining edges, exact clique computations on $G[N(u)\cap N(v)]$ are carried out to apply the edge reduction rule.
Whenever such a computation yields a clique, it is stored in a clique pool $\mathcal{C}_{pool}$ for later use in the EBCP-MILP formulation.

\begin{algorithm}[H]
    \caption{The preprocessing algorithm}
    \label{alg:reduce}
    \KwIn{Initial graph $G = (V, E)$, interdiction budget $k$}
    \KwOut{A reduced graph $G'$ and  a set of cliques $\mathcal{C}_{pool}$}
    \emph{Preprocess}$(G,k)$\\
   \Begin{
    $lb \gets \text{EstimateLB}(G,k)$ \;
    Recursively remove all vertices $u$ with $|N(u)| \leq lb - 2$, and all edges $\{u, v\}$ with $|N(u) \cap N(v)| \leq lb - 3$ from $G$ \; 
    \For{$u \in V$}{
        \If{$\chi(G[N(u)]) \leq lb-2$}{
            Remove $u$ from $G$ \;
        }
    }
    \For{$\{u,v\} \in E$}{
        \If{$\chi(G[N(u) \cap N(v)]) \leq lb - 3$}{
            Remove $\{u, v\}$ from $G$\;
        }
    }
    Initialize an empty clique pool $\mathcal{C}_{pool}$\;
    \For{$\{u,v\} \in E$}{
        $|C|\gets \mathrm{MaxClique}(G[N(u)\cap N(v)])$\;
        \eIf{$|C|\leq lb-3$}{
            $E\gets E\setminus\{ \{ u,v \} \}$\;
        }{
            $\mathcal{C}_{pool}\gets\mathcal{C}_{pool}\cup\{C\cup \{u,v\}\}$\;
        }
    }
    \Return{$G, \mathcal{C}_{pool}$}
    }
\end{algorithm}

\subsection{The Upper Bound Heuristic}

We now describe a heuristic procedure for computing an upper bound on $\eta(G,k)$.
The pseudocode of the proposed heuristic is given in Algorithm~\ref{UB}. The basic observation underlying this approach is that, for any edge set $F \subseteq E$ with $|F|\le |E|$, the clique number of the graph after removing $F$ provides an upper bound on $\eta(G,k)$.
Removing edges one by one and recomputing the maximum clique after each deletion would lead to prohibitively many maximum clique computations, while directly removing $k$ edges from a maximum clique of $G$ may easily miss high-quality solutions.
To balance solution quality and computational efficiency, we adopt an adaptive strategy that determines how many edges to remove at each iteration. Specifically, in the initial iteration $i=0$, we compute a maximum clique $C_0=\mathrm{MaxClique}(G)$ of the input graph and set the number of edges to be removed to $r=1$. In iteration $i\ge 1$, we first compute a maximum clique $C_{i+1}$ in the current graph. The value $r$ is updated according to the following rule. If the clique number does not decrease compared to the previous iteration, that is, $|C_{i+1}| = |C_{i}|$, we increase the number of edges to be removed by setting $r=\min(2r,|C_{i}|, k)$. Otherwise, if $|C_{i+1}|<|C_{i}|$, we reset $r$ to 1.

Subsequently, an edge set $F\subseteq E(C_i)$ of size $r$ is selected uniformly at random and removed from the current graph.
The remaining interdiction budget is updated to $k\gets k-r$, and the procedure continues until the budget is exhausted.

\begin{algorithm}[H] \label{UB}
\caption{The algorithms for estimating upper bound}
\label{alg:edge_removal_heuristic}
\KwIn{Graph $G = (V, E)$, integer $k\ge 0$}
\KwOut{An upper bound on $\eta(G,k)$.}
\emph{EstimateUB}($G,k$) \\
\Begin{
$i\gets 0$, $C_i \gets \mathrm{MaxClique}(G), r \gets 1$\;
\While{$k > 0$ }{
    $E_{\text{del}} \gets$ Randomly select $r$ edges from $E(C_i)$\;
    Remove edge set $E_{\text{del}}$ from graph $G$ \;
    $C_{i+1} \gets \mathrm{MaxClique}(G=(V, E \setminus E_{\text{del}}))$\;
    \If{$ \quad |C_{i+1}| = |C_{i}|$}{
        $r \gets \min(2r, |C_{i+1}|,k)$ \;
    }\Else{
        $r \gets 1$\;
    }    
    $k \gets k - r$, $i\gets i+1$ \;
}
\Return{$|C_{i-1}|$}
}
\end{algorithm}

\subsection{The Whole \textsc{RLCM} Algorithm}
\label{subsec_whole_algorithm}

We now integrate the previously introduced components into \textsc{RLCM}, an exact algorithm for solving the EICP, whose overall structure is summarized in Algorithm~\ref{alg:Frame}.
The algorithm starts with a preprocessing phase that reduces the input graph and initializes a clique pool $\mathcal{C}_{pool}$.
Using this pool, a partial EBCP-MILP (P-EBCP-MILP) is constructed by including Inequality~\eqref{EBCP2} only for cliques in $\mathcal{C}_{pool}$, together with additional permutation-based inequalities to strengthen the relaxation.

Starting from an initial upper bound $ub$, \textsc{RLCM} searches for the smallest parameter $p$ such that $\gamma(G,p)\le k$.
For a given $p$, the optimal value of the P-EBCP-MILP provides a valid lower bound on $\gamma(G,p)$.
If this bound exceeds $k$, the current value of $p$ can be safely discarded.
Otherwise, a branch-and-cut scheme with lazy constraint generation is applied.
Whenever the solution induces a clique of size at least $p$ in the residual graph, the corresponding blocking and permutation-based inequalities are separated and added to the model.
The process terminates when the smallest feasible value of $p$ is identified.

\medskip
\noindent\textbf{Remark}
In \textsc{RLCM}, the parameter $p$ is examined sequentially from $ub-1$ down to $lb$, rather than via a binary search strategy. The rationale is as follows. Suppose that, for a given $p = \overline{p}$, the P-EBCP-MILP constructed from the current clique pool  $\mathcal{C}_{pool}$ certified that $\eta(G,k) < \overline{p}$.
When subsequently considering smaller values $p < \overline{p}$, all cliques $C \in \mathcal{C}_{pool}$ necessarily satisfy $|C| > p$, and thus the same clique pool can be directly reused to build the P-EBCP-MILP.
In contrast, a binary search strategy may require evaluating values $p > \overline{p}$. In this case, the existing clique pool $\mathcal{C}_{pool}$ may contain cliques of size smaller than $p$, which are not useful for constructing valid constraints in the P-EBCP-MILP and would therefore necessitate additional clique separation. This makes the sequential strategy more efficient than binary search strategy in practice.

\begin{algorithm}[H]
    \caption{The two-stage algorithm, RLCM}
    \label{alg:Frame}
    \KwIn{Graph $G=(V,E)$, an integer $k$}
    \KwOut{$\eta(G,k)$}
    \emph{RLCM}$(G,k)$ \\
    \Begin{
        $G, \mathcal{C}_{pool} \gets \mathrm{Preprocess}(G,k)$\;
        Build EBCP-MILP by relaxing $\mathcal{C}$ as $\mathcal{C}_{pool}$, denoting the model as P-EBCP-MILP \;
        \For{ each $C \in \mathcal{C}_{pool}$}{ 
        Randomly generate an ordering $\overrightarrow{C}$ and its reversed ordering $\overrightarrow{C'}$. Then, add extra inequalities (\ref{sb:1})--(\ref{sb:4:variable}) which are constructed from $\overrightarrow{C}$ and $\overrightarrow{C'}$ to the P-EBCP-MILP model \;
        }
        $ub \gets \mathrm{EstimateUB}(G,k)$\;
        $p \gets ub-1$  \;
        \While{true}{
            Solve P-EBCP-MILP by the MILP solver, let the optimal value be $\hat{\gamma}(G, p)$ and the optimal solution be $\hat{x}$ \;
            \If{$\hat{\gamma}(G, p) \le  k$}{                    
                Let $F$ be the set of edges $e$ for which $\hat{x}_e=1$\;
                $C \gets \mathrm{MaxClique}(G=(V,E\setminus F))$\;
                \If{$|C| < p$}{
                    $p \gets |C| - 1$\;
                    Rebuild the P-EBCP-MILP where $\mathcal{C}$ is $\mathcal{C}_{pool}$
                }
                Add Inequality $\sum_{e \in E(C)} x_e \geq \gamma_{clq}(|C|,p)$ to P-EBCP-MILP \;        
                Randomly generate an ordering $\overrightarrow{C}$ and its reversed ordering $\overrightarrow{C'}$ \;
                Add extra inequalities (\ref{sb:1})--(\ref{sb:4:variable}) which are constructed from $\overrightarrow{C}$ and $\overrightarrow{C'}$ to the P-EBCP-MILP model \;
            }\Else{
                \Return{$p + 1$}    
            }
        }
        
    }
\end{algorithm}

\section{Experiment}
\subsection{Experiment Setup}
In this section, we evaluate the performance of our proposed algorithm RLCM. 
All our experiments were conducted on a server running Ubuntu 22.04, equipped with an Intel Xeon\textsuperscript{\textregistered} Platinum 8360Y CPU (2.40\,GHz) and 1TB RAM. 
The algorithm was implemented in \texttt{C++}, compiled using \texttt{g++} 10.5.0 with the \texttt{-O3} optimization flag. 
All mixed-integer linear programs were solved using IBM CPLEX 22.1.1 with a single thread and default settings.

\subsection{Benchmark Algorithms and Datasets}

We compare the proposed algorithm \textsc{RLCM} with two existing approaches for solving the EICP:
\begin{itemize}
    \item  \textsc{BILEVEL}, a state-of-the-art general-purpose mixed-integer bilevel optimization solver \citep{fischetti2017new,fischetti2018dynamic}.
    \item \textsc{EDGE-INTER}, the most recent exact algorithm specifically designed for the EICP \citep{furini2021branch}.
\end{itemize}

In addition, to ensure a fair comparison, we apply our preprocessing procedure to the input graph before invoking \textsc{BILEVEL}. Consequently, \textsc{BILEVEL} and \textsc{RLCM} are executed on the same reduced instances. We evaluate the algorithms on three benchmark datasets with different structural properties and from different sources.

\begin{itemize}
  \item \textbf{DIMACS2:} We consider 16 graphs from the Second DIMACS Implementation Challenge on the Maximum Clique Problem, each with $|V| = 200$. These instances have been widely used in the literature on maximum clique algorithms~\citep{chang2019efficient,wu2015review,zhou2021improving} and were also tested in \textsc{EDGE-INTER}.

  \item \textbf{Large sparse graphs:} 
    We use 95 large sparse graphs from the Network Repository \citep{networkrepository}, each containing at most $1\,000\,000$ edges. These graphs cover a wide range of real-world networks, including co-authorship, citation, street, clustering, geometric, and numerical simulation networks.

    \item \textbf{Random graphs:} We generated Erd\H{o}s–R\'enyi graphs with $|V| \in \{25, 50, 75, 100\}$ and edge density $\rho \in \{0.1, 0.2, 0.3, 0.4, 0.5, 0.6, 0.7, 0.8,$ $0.9, 0.95\}$. For each $(|V|,\rho)$ combination, 5 graphs are generated with different seeds, resulting in  a total of $4\times10\times5=200$ synthetic graphs.
\end{itemize}

For each DIMACS2 instance, we test five interdiction budgets $k$ in the set $\{10, 15, 20, 25, 30\}$. 
For the large real-world graphs, we consider both five fixed values $k\in\{10,15,20,25,30\}$ and five values proportional to the number of edges, namely $k=\lceil c \cdot |E| \rceil$ for $c \in \{0.0001, 0.0005, 0.0010, 0.0015, 0.0020\}$.  
For each random graph, we evaluate twelve values of $k$ relative to $|E|$, given by $k = \lceil c \cdot |E| \rceil$ for $c \in \{0.01, \ldots, 0.10, 0.15, 0.20\}$.
Each experiment is run with a time limit of 600 seconds. When computing average running times, instances that time out are counted as reaching the time limit, \textit{i.e.}, 600 seconds. Finally, we note that the experimental settings for the DIMACS2 and random graph instances are consistent with those used in \textsc{EDGE-INTER} \citep{furini2021branch}.





\subsection{Results with DIMACS2}
\begin{table}[H]
    \centering
    \scriptsize
    \setlength{\tabcolsep}{3pt} 
    \caption{Comparison of CPU time (in seconds) for \textsc{EDGE-INTER} and \textsc{RLCM} on DIMACS2.}
    \label{tab:D2-comparing}
    
    \begin{tabular}{@{}l c *{15}{c}@{}}
    \toprule
    & & 
    \multicolumn{3}{c}{$k=10$} & 
    \multicolumn{3}{c}{$k=15$} & 
    \multicolumn{3}{c}{$k=20$} & 
    \multicolumn{3}{c}{$k=25$} & 
    \multicolumn{3}{c}{$k=30$} \\
    \cmidrule(lr){3-5} \cmidrule(lr){6-8} \cmidrule(lr){9-11} \cmidrule(lr){12-14} \cmidrule(l){15-17}
    \cmidrule(r){3-5} \cmidrule(lr){6-8} \cmidrule(lr){9-11} \cmidrule(lr){12-14} \cmidrule(l){15-17}
    Instance & {$\omega(G)$} & Ed-In. & RLCM & $\eta$ & Ed-In. & RLCM & $\eta$ & Ed-In. & RLCM & $\eta$ & Ed-In. & RLCM & $\eta$ & Ed-In. & RLCM & $\eta$ \\
    \midrule
    brock200\_1 & 21 & 33.1 & \textbf{13.5} & 19 & \textbf{21.3} & 27.1 & 19 & t.l & \textbf{19.9} & 19 & - & \textbf{481.8} & 19 & - & t.l & - \\
    brock200\_2 & 12 & 0.6 & 0.6 & 10 & 0.8 & \textbf{0.7} & 10 & 2.7 & \textbf{0.7} & 10 & - & \textbf{8.7} & 9 & - & \textbf{8.5} & 9 \\
    brock200\_3 & 15 & \textbf{1.4} & 2.3 & 13 & 6.2 & \textbf{1.9} & 13 & 9.8 & \textbf{2.0} & 13 & - & \textbf{2.1} & 13 & - & \textbf{396.9} & 12 \\
    brock200\_4 & 17 & \textbf{4.0} & 4.1 & 15 & 9.1 & \textbf{4.4} & 15 & \textbf{19.0} & 21.4 & 14 & - & \textbf{21.0} & 14 & - & \textbf{21.2} & 14 \\
    c-fat200-1 & 12 & 0.1 & 0.1 & 11 & 0.1 & 0.1 & 10 & 0.1 & 0.1 & 10 & - & \textbf{0.1} & 10 & - & \textbf{0.1} & 10 \\
    c-fat200-2 & 24 & \textbf{0.2} & 0.3 & 22 & \textbf{0.2} & 0.4 & 21 & \textbf{0.2} & 0.6 & 20 & - & \textbf{0.4} & 20 & - & \textbf{2.2} & 19 \\
    c-fat200-5 & 58 & t.l & \textbf{5.9} & 55 & t.l & \textbf{9.4} & 53 & t.l & \textbf{8.4} & 52 & - & \textbf{16.3} & 50 & - & \textbf{20.8} & 49 \\
    gen200\_p0.9\_44 & 44 & t.l & t.l & - & t.l & t.l & - & $\textbf{21.9}$ & t.l & 16 & - & t.l & - & - & t.l & - \\
    gen200\_p0.9\_55 & 55 & t.l & \textbf{26.7} & 45 & t.l & t.l & - & ${t.l}$ & t.l & - & - & t.l & - & - & t.l & - \\
    san200\_0.7\_1 & 30 & \textbf{5.5} & 7.7 & 20 & 119.4 & \textbf{65.2} & 17 & ${t.l}$ & \textbf{25.9} & 17 & - & \textbf{26.6} & 17 & - & \textbf{30.5} & 17 \\
    san200\_0.7\_2 & 18 & \textbf{0.7} & 2.5 & 15 & 2.7 & \textbf{2.6} & 15 & ${t.l}$ & \textbf{2.6} & 15 & - & t.l & - & - & \textbf{306.6} & 14 \\
    san200\_0.9\_1 & 70 & \textbf{66.2} & 121.7 & 60 & t.l & t.l & - & ${t.l}$ & t.l & - & - & t.l & - & - & t.l & - \\
    san200\_0.9\_2 & 60 & \textbf{66.5} & 70.1 & 50 & t.l & \textbf{39.2} & 45 & $\textbf{55.6}$ & t.l & 15$^?$ & - & t.l & - & - & t.l & - \\
    san200\_0.9\_3 & 44 & t.l & \textbf{181.5} & 36 & t.l & t.l & - & ${t.l}$ & t.l & - & - & t.l & - & - & t.l & - \\
    sanr200\_0.7 & 18 & 15.4 & \textbf{6.0} & 17 & \textbf{20.0} & 34.7 & 16 & ${t.l}$ & \textbf{25.5} & 16 & - & \textbf{26.2} & 16 & - & \textbf{31.3} & 16 \\
    sanr200\_0.9 & 42 & t.l & t.l & - & t.l & t.l & - & ${t.l}$ & t.l & - & - & t.l & - & - & t.l & - \\
    \midrule
    \multicolumn{2}{l}{\#opt} & 
    11 & \textbf{14} & &
    9 & \textbf{11} & &
    7 & \textbf{10} & &
    - & \textbf{9} & &
    - & \textbf{9} \\
    \bottomrule
    \end{tabular}
\end{table}

Table~\ref{tab:D2-comparing} reports the computational time of \textsc{RLCM} and \textsc{EDGE-INTER}. 
The solver \textsc{BILEVEL} is omitted from the table, as it failed to solve any EICP instance within the 600-second time limit.

The source code of \textsc{EDGE-INTER}, implemented in \texttt{C}, could not be successfully compiled on our platform. Despite communication with the authors, we were unable to obtain a complete and usable version of the code. After communicating with the authors, we were unable to obtain the complete source codes. Consequently, the experimental results for \textsc{EDGE-INTER} are taken directly from \citep{furini2021branch}.
The experimental settings in \citep{furini2021branch} are consistent with ours. Moreover, their experiments were conducted on a machine equipped with a 3.40 GHz 8-core Intel i7-3770 processor, which has approximately 1.4 times higher clock frequency than our platform. We therefore believe that the reported results of \textsc{EDGE-INTER} can be fairly and directly compared with our results.

In Table~\ref{tab:D2-comparing}, $\eta$ denotes the optimal value $\eta(G,k)$. "t.l" is displayed when the time limit was exceeded, while "-"  denote instances for which \textsc{EDGE-INTER} could not be evaluated due to unavailable source code or for which no algorithm was able to compute $\eta(G,k)$ within the time limit. 
For several \textsc{EDGE-INTER} entries at $k=20$, we use the symbol "?" to mark values that appear inconsistent in the original paper. For example, for instance \texttt{san200\_0.9\_2}, \textsc{EDGE-INTER} reports $\eta(G,15)=45$ and $\eta(G,20)=15$, which is impossible, as removing five additional edges cannot reduce a clique of size 45 to size 15.
The last row of the table reports the total number of EICP instances solved by each algorithm.

From Table~\ref{tab:D2-comparing}, it is clear that \textsc{RLCM} solves more instances than the competing algorithms for every tested value of $k$.
In many cases, \textsc{EDGE-INTER} fails to return a solution within the time limit, whereas \textsc{RLCM} successfully solves the corresponding instances within 600 seconds.
Conversely, there are only two instances, \texttt{gen200\_p0.9\_44} and \texttt{san200\_0.9\_2} with 
$k=20$, that are reportedly solved by \textsc{EDGE-INTER} but not by \textsc{RLCM}. 
Notably, in the latter case, the solution reported by \textsc{EDGE-INTER} appears questionable.
Overall, \textsc{RLCM} remains scalable on these hard clique instances even for larger budgets like  $k=25$ or $30$.
A closer inspection of the solving process reveals that the preprocessing stage is unable to reduce the size of the input graphs for this dataset. Consequently, it is not surprising that \textsc{BILEVEL} fails to handle these instances effectively.

\subsection{Results with Large Sparse Graphs}

\begin{figure*}[!h]
    \centering
     \subfigure[The running time at different fixed values of $k$ for each instance. ]{
        \includegraphics[width=0.46\textwidth]{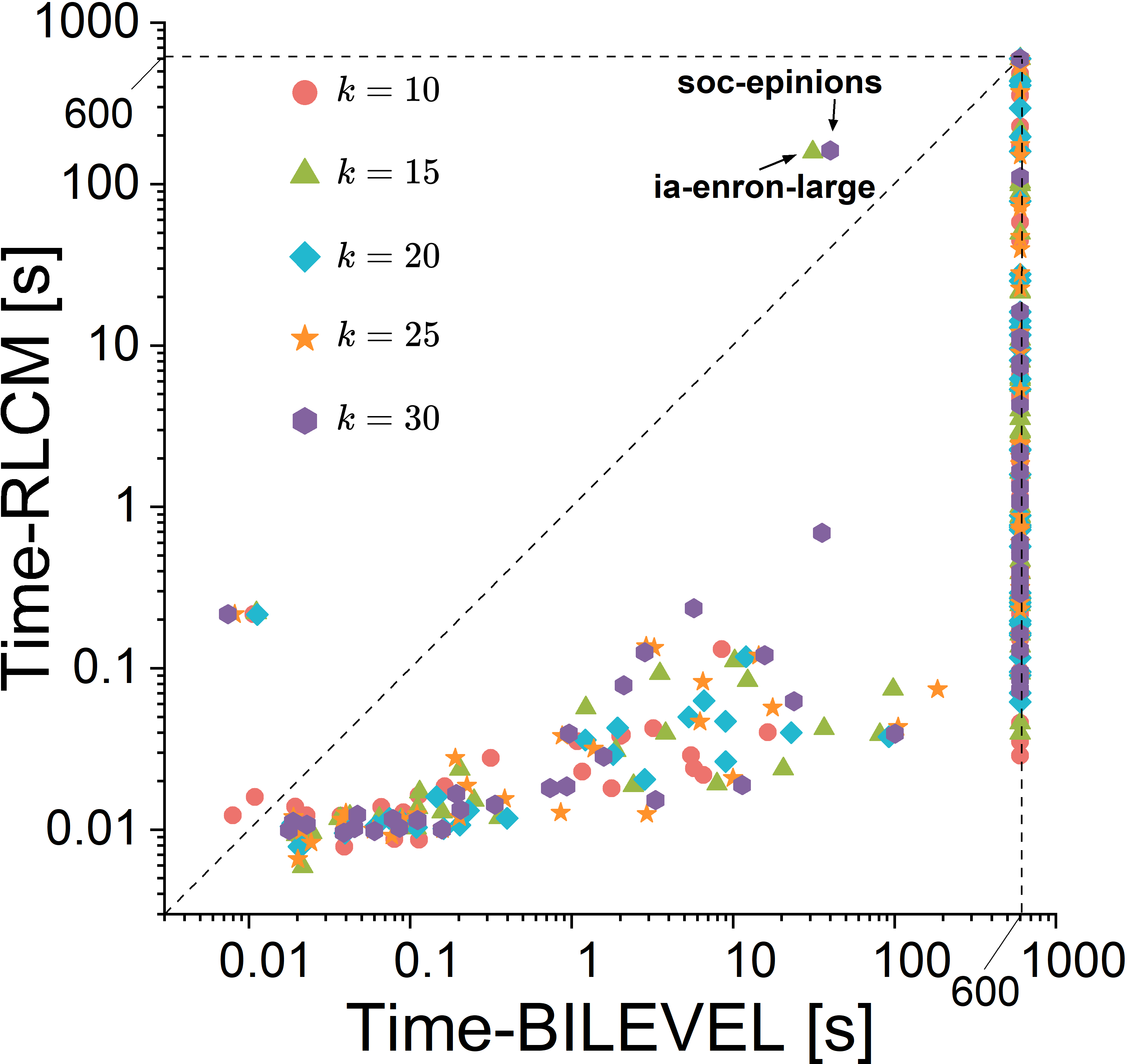}
        \label{fig:Time-Time-k}
    }
    \hfill
    \subfigure[The running time at different values of $k=\lceil c \cdot |E| \rceil$ for each instance.]{
        \includegraphics[width=0.46\textwidth]{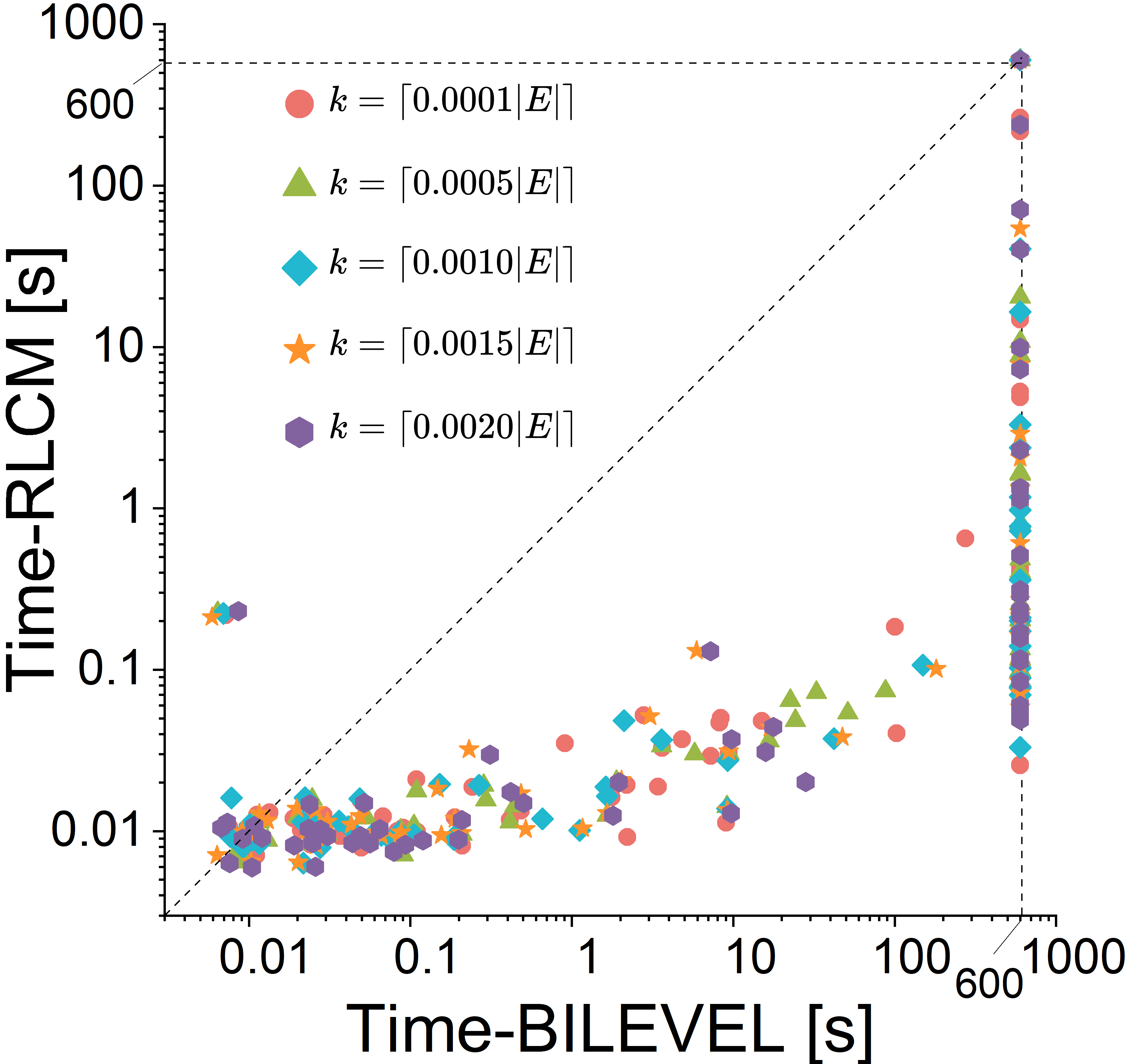}
        \label{fig:Time-Time-d}
    }
    
    \vspace{0.5cm}
    
    \subfigure[The number of optimal solutions at different fixed values of $k$ for each instance.]{
        \includegraphics[width=0.46\textwidth]{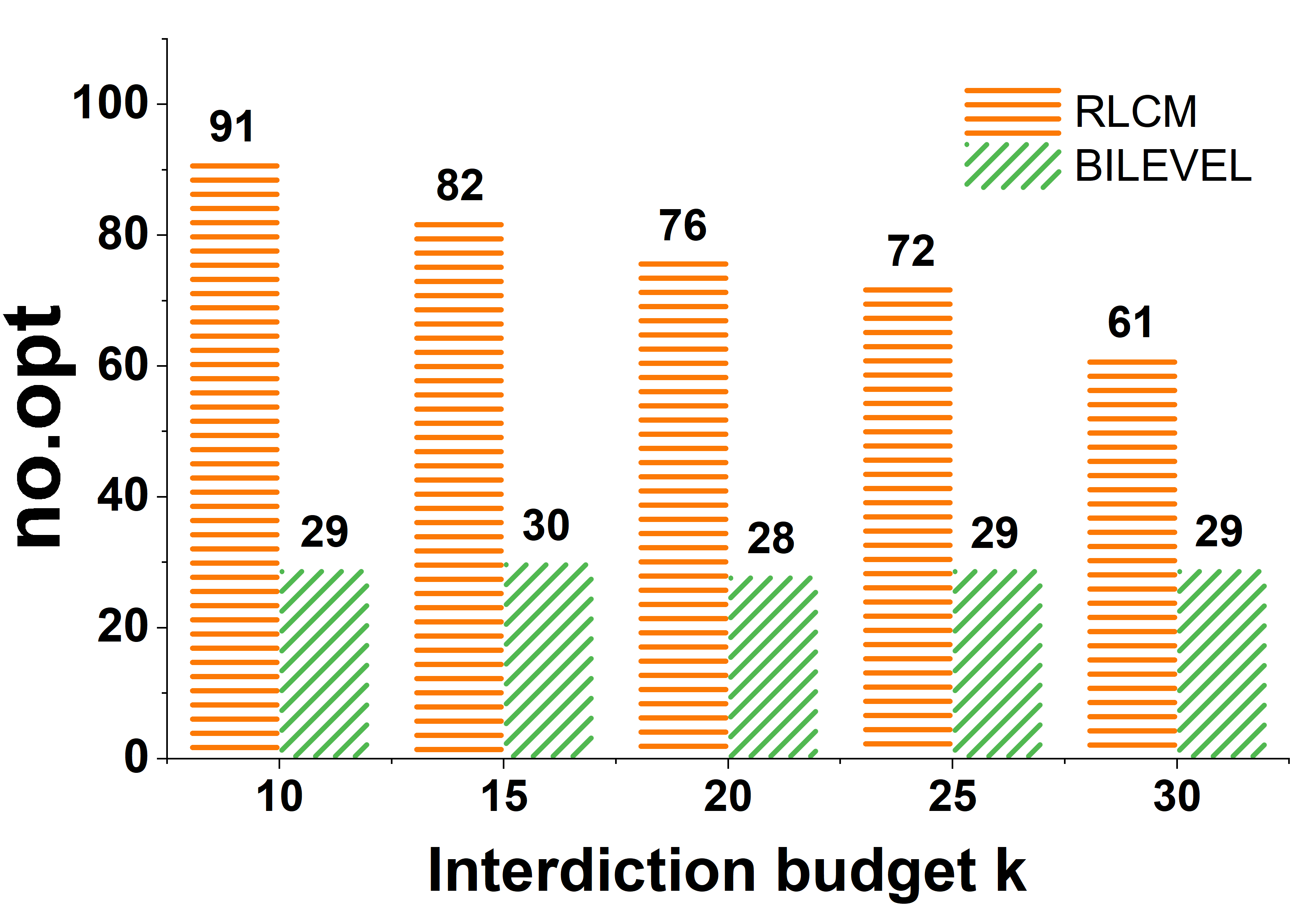}
        \label{fig:opt-k}
    }
    \hfill
    \subfigure[The number of optimal solutions at different values of $k=\lceil c \cdot |E| \rceil$ for each instance.]{
        \includegraphics[width=0.46\textwidth]{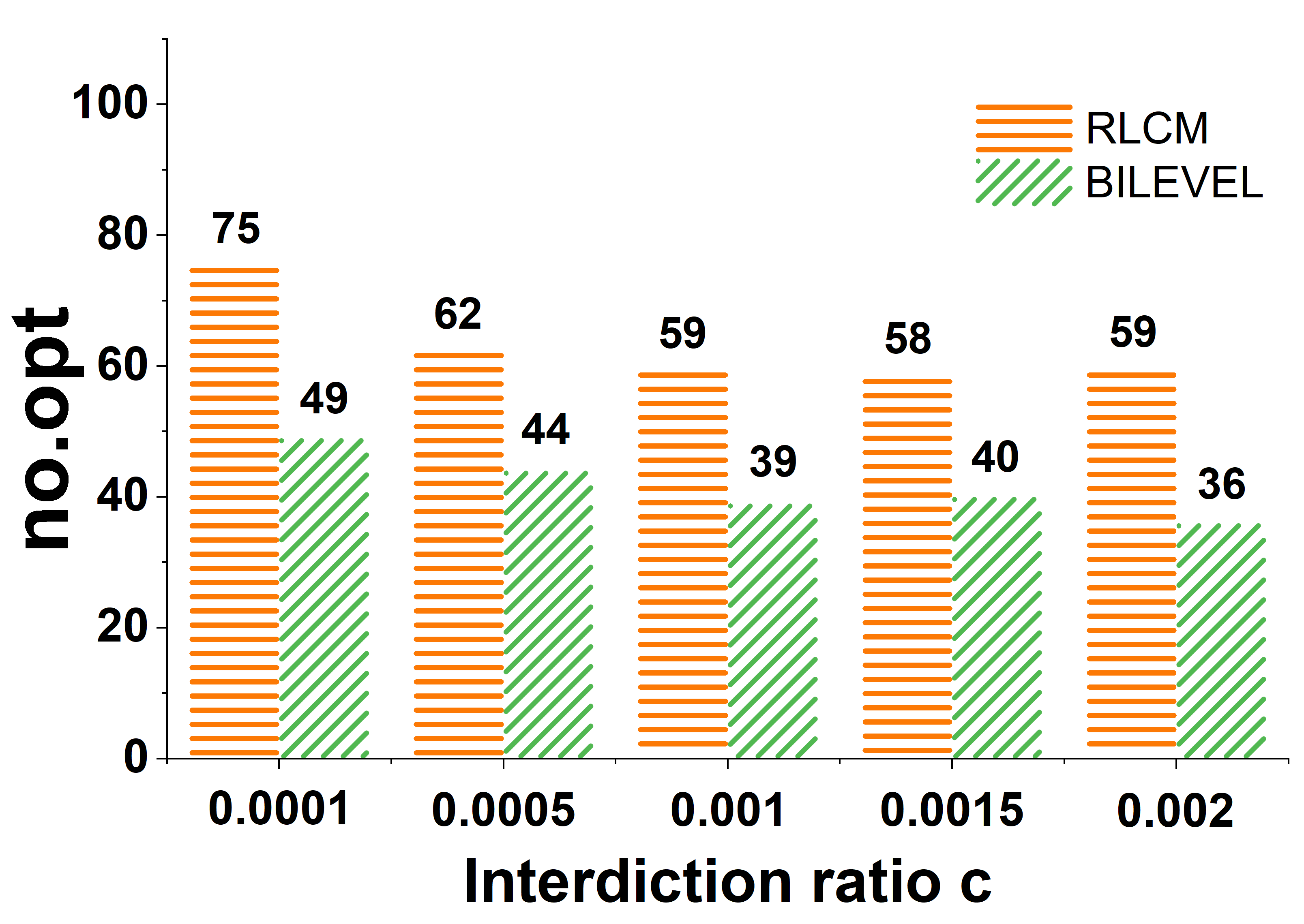}
        \label{fig:opt-d}
    }
    
    \caption{Performance of two algorithms across different values of $k$.}
    \label{fig:combined-results}
\end{figure*}

Figure~\ref{fig:combined-results} compares the performance of \textsc{RLCM} and \textsc{BILEVEL} on large sparse graphs for different values of $k$. 
As discussed earlier, \textsc{EDGE-INTER} is not included in this comparison due to the unavailability of its complete source codes and corresponding experimental data.
In Figures~\ref{fig:combined-results}(a) and~(b), each point corresponds to an EICP instance $(G,k)$.
The coordinates of a point represent the running time of the two algorithms on the same instance. 
Points lying below the diagonal indicate instances for which \textsc{RLCM} is faster than \textsc{BILEVEL}, while points above the diagonal indicate the opposite.
Figures~\ref{fig:combined-results}(c) and~(d) summarize the number of instances solved by each algorithm for different values of $k$.

Across all tested values of $k$, there are 695 instances that are solved by both \textsc{RLCM} and \textsc{BILEVEL} within the time limit. 
Among these instances, \textsc{RLCM} outperforms \textsc{BILEVEL} on 664 cases in terms of computational time.
The advantage of \textsc{BILEVEL} is mainly observed on relatively easy instances that can be solved in approximately 0.1 seconds.
In addition, \textsc{BILEVEL} fails on 180 instances due to excessive memory consumption, even after applying the same preprocessing procedure. In contrast, \textsc{RLCM} is able to solve 113 of these 180 instances to optimality.

\subsection{Evaluating the Reduction, Upper Bound Estimation and Enhanced Inequalities}



Graph reduction, upper bound estimation, and branch-and-cut with enhanced inequalities are key components of the \textsc{RLCM} framework. In this subsection, we evaluate the individual impact of these components through an ablation study. We consider three ablated variants of \textsc{RLCM}:
\begin{itemize}
    \item ~\textsc{NRed}, which disables all graph reduction rules;
    \item ~\textsc{NUB}, which removes the upper-bound estimation procedure; and
    \item ~\textsc{NCon}, which omits the additional Inequalities~\eqref{sb:1}–\eqref{sb:4:variable} introduced in Section~\ref{subsection_ineq_max_clique} during cut generation.
\end{itemize}


Figure~\ref{fig:ablation_time} reports the results of this study on the set of large sparse graphs because the graph reduction mostly works in large sparse graphs. 
The figure presents box plots of the running times for each algorithmic variant, where each individual EICP instance is represented by a dot. In addition, the number of instances solved by each variant within the time limit is reported at the top of the figure.

The results clearly indicate that the variant without reduction, \textsc{NRed}, shows substantially higher CPU time and consistently solves fewer instances to optimality than the other variants. 
This observation highlights the critical role of graph reduction in the \textsc{RLCM} framework when solving large sparse graphs.

\begin{figure}[H]
    \centering
    \includegraphics[width=0.98\linewidth]{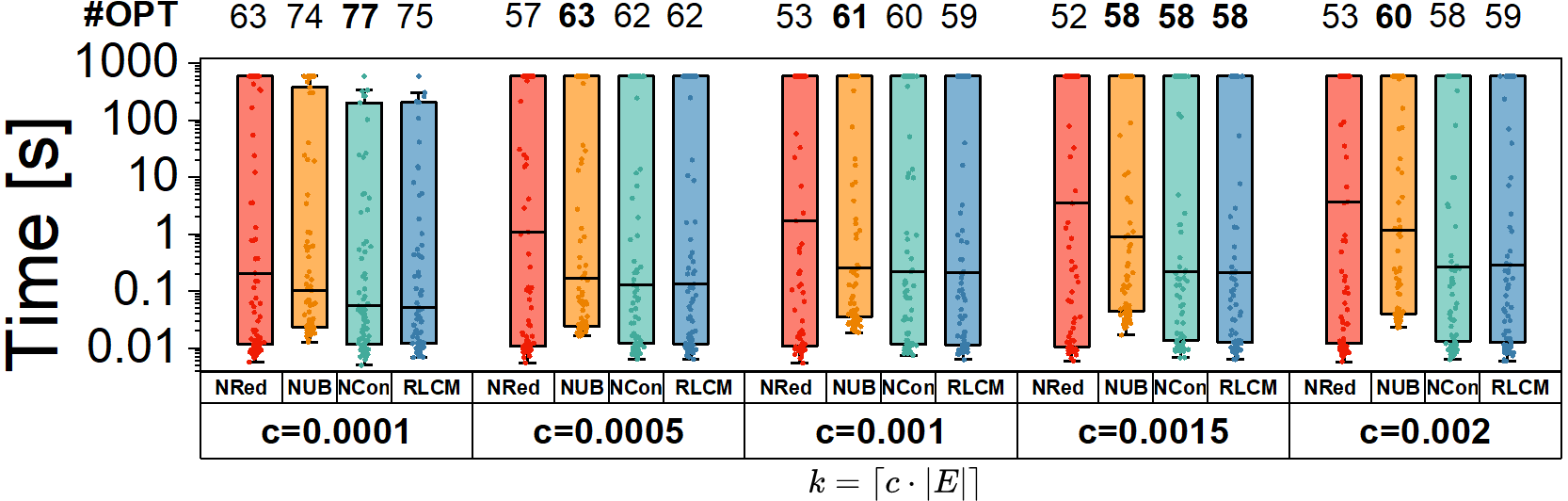}
    \caption{The statistical CPU time of \textsc{RLCM} and its variants on large sparse graphs.}
    \label{fig:ablation_time}
\end{figure}


To further illustrate the effectiveness of the preprocessing stage, Figure~\ref{fig:reduction} reports the proportion of vertices and edges removed from the large sparse graphs during graph reduction.
Across this dataset, the reduction achieves substantial reduction, leading to a significant simplification of the problem instances.
We observe, however, that the effectiveness of the reduction decreases when the interdiction budget $k$ increases. We conjecture that this behavior is primarily due to the lower bound on $\eta(G,k)$ becoming progressively less tight for larger values of $k$.

\begin{figure}[htbp]
    \centering
    \subfigure[The percentage of edges removed for each instance at different $k=\lceil c \cdot |E| \rceil$.]{
        \includegraphics[width=0.46\textwidth]{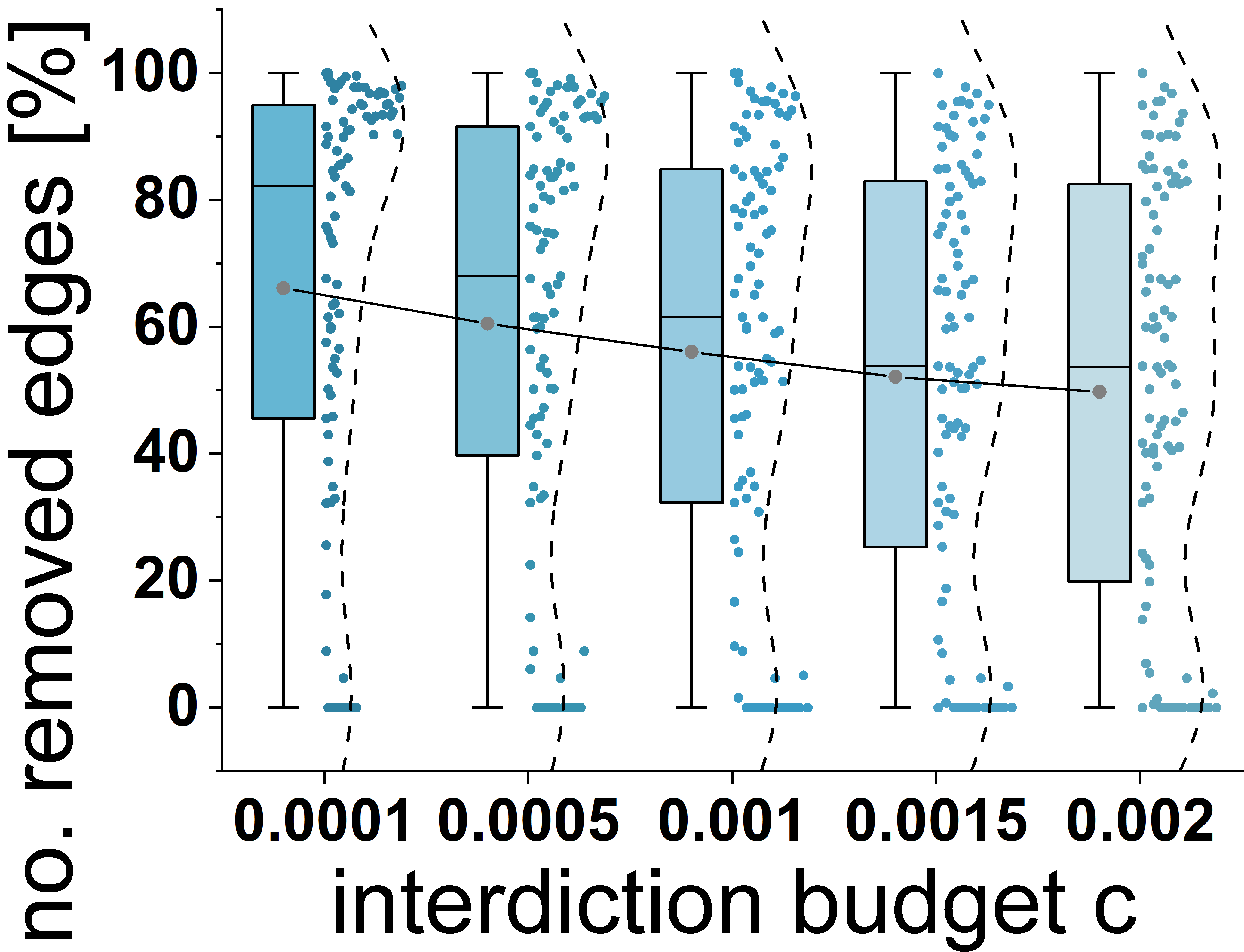}
        \label{fig:removed_edges-d}
    }
    \hfill
    \subfigure[The percentage of vertices removed for each instance at different $k=\lceil c \cdot |E| \rceil$.]{
        \includegraphics[width=0.46\textwidth]{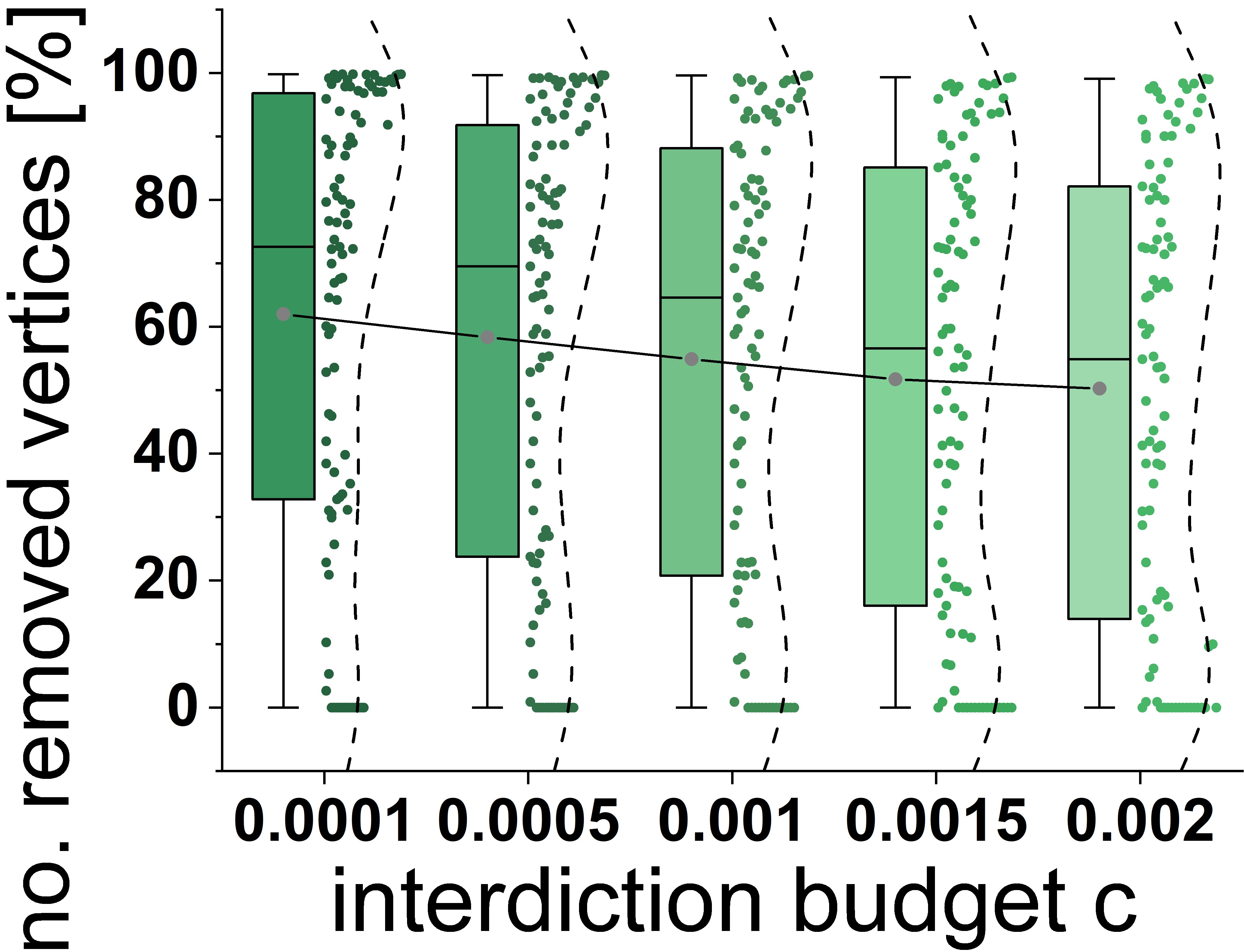}
        \label{fig:removed_vertices-d}
    }
    \caption{Reduction effectiveness on large sparse graphs. Note that each dot represents a graph. }
    \label{fig:reduction}
\end{figure}

\begin{table*}[h]
  \centering
  \scriptsize
  \caption{Comparison of computational time of \textsc{RLCM} and \textsc{NCon} for different $k$ on DIMACS2 dataset.} %
  \label{tab: comparation of R and NC}
  \setlength{\tabcolsep}{2pt} 
    \begin{tabular}{@{}l r *{15}{r} @{}} 
    \toprule
    & &
    \multicolumn{3}{c}{$k=10$} & 
    \multicolumn{3}{c}{$k=15$} & 
    \multicolumn{3}{c}{$k=20$} & 
    \multicolumn{3}{c}{$k=25$} & 
    \multicolumn{3}{c}{$k=30$} \\
    \cmidrule(lr){3-5} \cmidrule(lr){6-8} \cmidrule(lr){9-11} \cmidrule(lr){12-14} \cmidrule(l){15-17}
    Instance & {$\omega(G)$} & {NCon} & {RLCM} & {$\eta$} 
    & {NCon} & {RLCM} & {$\eta$} 
    & {NCon} & {RLCM} & {$\eta$} 
    & {NCon} & {RLCM} & {$\eta$} 
    & {NCon} & {RLCM} & {$\eta$} \\
    \midrule
    brock200\_1 & 21 & 13.7 & \textbf{13.5} & 19 & \textbf{20.1} & 27.1 & 19 & \textbf{18.2} & 19.9 & 19 & {t.l} & \textbf{481.8} & 19 & {t.l} & {t.l} & {-} \\
    brock200\_2 & 12 & 0.6 & 0.6 & 10 & 0.7 & 0.7 & 10 & 0.7 & 0.7 & 10 & \textbf{7.6} & 8.7 & 9 & \textbf{5.0} & 8.5 & 9 \\
    brock200\_3 & 15 & 2.7 & \textbf{2.3} & 13 & 1.9 & 1.9 & 13 & 2.0 & 2.0 & 13 & 2.1 & 2.1 & 13 & {t.l} & \textbf{396.9} & 12 \\
    brock200\_4 & 17 & 4.2 & \textbf{4.1} & 15 & 4.5 & \textbf{4.4} & 15 & 23.4 & \textbf{21.4} & 14 & 22.2 & \textbf{21.0} & 14 & 30.4 & \textbf{21.2} & 14 \\
    c-fat200-1 & 12 & 0.1 & 0.1 & 11 & 0.1 & 0.1 & 10 & 0.1 & 0.1 & 10 & 0.1 & 0.1 & 10 & 0.1 & 0.1 & 10 \\
    c-fat200-2 & 24 & 0.3 & 0.3 & 22 & \textbf{0.3} & 0.4 & 21 & \textbf{0.4} & 0.6 & 20 & 0.6 & \textbf{0.4} & 20 & \textbf{0.8} & 2.2 & 19 \\
    c-fat200-5 & 58 & 6.1 & \textbf{5.9} & 55 & \textbf{6.4} & 9.4 & 53 & \textbf{7.8} & 8.4 & 52 & \textbf{16.1} & 16.3 & 50 & 153.9 & \textbf{20.8} & 49 \\
    gen200\_p0.9\_44 & 44 & {t.l} & {t.l} & {-} & {t.l} & {t.l} & {-} & {t.l} & {t.l} & {-} & {t.l} & {t.l} & {-} & {t.l} & {t.l} & {-} \\
    gen200\_p0.9\_55 & 55 & 54.2 & \textbf{26.7} & 45 & {t.l} & {t.l} & {-} & {t.l} & {t.l} & {-} & {t.l} & {t.l} & {-} & {t.l} & {t.l} & {-} \\
    san200\_0.7\_1 & 30 & 17.6 & \textbf{7.7} & 20 & 81.2 & \textbf{65.2} & 17 & 28.3 & \textbf{25.9} & 17 & 34.5 & \textbf{26.6} & 17 & 33.5 & \textbf{30.5} & 17 \\
    san200\_0.7\_2 & 18 & 3.0 & \textbf{2.5} & 15 & \textbf{2.5} & 2.6 & 15 & 2.6 & 2.6 & 15 & {t.l} & {t.l} & {-} & \textbf{191.4} & 306.6 & 14 \\
    san200\_0.9\_1 & 70 & \textbf{16.3} & 121.7 & 60 & {t.l} & {t.l} & {-} & {t.l} & {t.l} & {-} & {t.l} & {t.l} & {-} & {t.l} & {t.l} & {-} \\
    san200\_0.9\_2 & 60 & 367.8 & \textbf{70.1} & 50 & {t.l} & \textbf{39.2} & 45 & {t.l} & {t.l} & {-} & {t.l} & {t.l} & {-} & {t.l} & {t.l} & {-} \\
    san200\_0.9\_3 & 44 & 465.5 & \textbf{181.5} & 36 & {t.l} & {t.l} & {-} & {t.l} & {t.l} & {-} & {t.l} & {t.l} & {-} & {t.l} & {t.l} & {-} \\
    sanr200\_0.7 & 18 & \textbf{5.9} & 6.0 & 17 & \textbf{29.8} & 34.7 & 16 & 33.8 & \textbf{25.5} & 16 & \textbf{20.6} & 26.2 & 16 & \textbf{27.9} & 31.3 & 16 \\
    sanr200\_0.9 & 42 & {t.l} & {t.l} & {-} & {t.l} & {t.l} & {-} & {t.l} & {t.l} & {-} & {t.l} & {t.l} & {-} & {t.l} & {t.l} & {-} \\
    \bottomrule
  \end{tabular}
\end{table*}

Regarding the upper bound estimation, the variant without this component, \textsc{NUB}, solves more instances for all the $k$ values tested  except $k=\lceil 0.0001|E| \rceil$.
However, a closer examination of the per-instance running times reveals a different picture.
Although \textsc{NUB} solves slightly more instances overall, it is generally less efficient on a per-instance basis.
In particular, among the instances solved by both algorithms, \textsc{NUB} requires more time than \textsc{RLCM} on 290 instances, whereas \textsc{RLCM} is slower on only 28 instances.

We next compare \textsc{RLCM} with the variant that omits the additional inequalities, denoted by \textsc{NCon}.
On the large sparse graph dataset, the two algorithms exhibit comparable performance in terms of both the number of solved instances and median computational time. 
Nevertheless, this observation should not be interpreted as diminishing the importance of the additional inequalities, since their impact is more clear on dense and challenging instances.
To highlight this effect, Table~\ref{tab: comparation of R and NC} reports the running times of \textsc{RLCM} and \textsc{NCon} on the dense DIMACS2 graphs. 
On these instances, \textsc{RLCM} demonstrates a clear advantage, solving three additional EICP instances that \textsc{NCon} fails to solve within the time limit, namely \texttt{san200\_0.9\_2} for $k=15$, \texttt{brock200\_1} for $k=25$, and \texttt{brock200\_3} for $k=30$.


\subsection{Results with Random Graphs}

Table~\ref{tab:ER} summarizes the experimental results on the Erd\H{o}s–R\'enyi random graphs, grouped by different values of $k$.
Each group contains 200 instances with varying numbers of vertices and edge densities. 
For each group, we report the average number of instances for which an optimal solution is obtained (\#opt), as well as the average running time required to compute the optimal solution.
Overall, \textsc{RLCM} consistently outperforms the other two algorithms on this dataset, both in terms of the number of solved instances and running time. An exception is observed when 
$k \ge \lceil0.10\cdot|E|\rceil$, where \textsc{EDGE-INTER} exhibits a slight advantage with respect to \#opt.


\begin{table}[H]
  \centering
  \scriptsize
  \caption{Algorithm performance comparison grouped by different $k$, where $k = \lceil c \cdot |E| \rceil$.  
}
  \label{tab:ER}

  \setlength{\tabcolsep}{15pt} 
  \begin{tabular}{@{}r 
                  *{3}{r@{ }r}@{}} 
    \toprule
    {$k= \lceil c \cdot |E| \rceil$} & 
    \multicolumn{2}{c}{\textbf{\sc BILEVEL}} & 
    \multicolumn{2}{c}{\textbf{\sc EDGE-INTER}} & 
    \multicolumn{2}{c}{\textbf{\sc RLCM}} \\
    \cmidrule(lr){2-3} \cmidrule(lr){4-5} \cmidrule(lr){6-7} 
    {$c$} & {\#opt} & {time (s)} & {\#opt} & {time (s)} & {\#opt} & {time (s)} \\
    \midrule
    0.01  & 18 & 340.8 & 26 & 210.5 & \textbf{32.0} & \textbf{131.1} \\
    0.02  & 15.2 & 386.2 & 25 & 234.8 & \textbf{27.0} & \textbf{197.1} \\
    0.03  & 13.4 & 409.3 & 22 & 282.2 & \textbf{24.4} & \textbf{241.7} \\
    0.04  & 11.6 & 435.0 & 22 & 287.6 & \textbf{24.8} & \textbf{236.7} \\
    0.05  & 11.4 & 437.2 & 22 & 300.2 & \textbf{24.0} & \textbf{242.3} \\
    0.06  & 9.6 & 462.5 & 18 & 348.7 & \textbf{21.6} & \textbf{280.8} \\
    0.07  & 8.8 & 474.0 & 17 & 362.8 & \textbf{20.0} & \textbf{302.3} \\
    0.08  & 8.6 & 476.7 & 18 & 357.0 & \textbf{18.8} & \textbf{320.1} \\
    0.09  & 9.2 & 468.7 & 17 & 364.7 & \textbf{18.8} & \textbf{325.2} \\
    0.10  & 9 & 474.1 & \textbf{19} & 345.4 & 18.4 & \textbf{333.4} \\
    0.15  &  8 & 485.1 & \textbf{17} & 397.1 & 16 & \textbf{367.0} \\
    0.20  & 9 & 475.2 & \textbf{17} & 383.1 & 16.6 & \textbf{362.3} \\
    \bottomrule
  \end{tabular}
  
  \vspace{4pt}
  \footnotesize
  \raggedright
\end{table}

\section{Conclusion and Future Work}

In this paper, we investigated mixed-integer linear programming formulations and exact solution methods for the Edge Interdiction Clique Problem (EICP), which seeks to remove at most 
$k$ edges from a graph in order to minimize its clique number. The problem is computationally challenging and belongs to a complexity class beyond NP.

We reformulated the EICP as a sequence of parameterized Edge Blocker Clique Problems (EBCPs) and developed strengthened MILP formulations by introducing new valid inequalities that exploit different clique orderings. Building on these formulations, we proposed \textsc{RLCM}, a two-stage exact algorithm that integrates graph reduction techniques, upper-bound estimation, and an iterative branch-and-cut framework. Together, these components yield a practically tractable approach for solving the EICP.

Extensive computational experiments demonstrate that the proposed method is both efficient and robust across graphs with diverse structural properties. On the DIMACS2 benchmark instances, \textsc{RLCM} solves more EICP instances within the 600-second time limit than existing approaches and resolves several previously unsolved cases. 
On large sparse real-world networks, the algorithm exhibits strong scalability, while general-purpose bilevel solvers often fail due to excessive memory consumption. 
Additional experiments on random graphs further confirm the adaptability of \textsc{RLCM} across varying graph densities and problem sizes. 
Finally, an ablation study highlights the individual contributions of graph reduction, upper-bound estimation, and the proposed inequalities to the overall performance of the algorithm.

There are several promising directions for future research.
First, interdiction problems on alternative notions of dense subgraphs merit further investigation.
Cliques are often too restrictive for modeling real systems, and relaxed structures such as $s$-plexes~\citep{zhou2017frequency} and $s$-bundles~\citep{zhou2022effective} provide more realistic characterizations of dense communities.
Recent work by~\citep{zhong2024interdicting} has initiated this line of study by considering quasi-cliques—subgraphs with edge density above a given threshold—and it would be natural to examine broader relaxations within the same framework.
Another direction is to apply the reduction and modeling ideas developed here to temporal networks, where dense structures evolve over time.
In addition, interdiction actions could be extended beyond edge deletions or additions; the impact of vertex-based modifications, including weighted vertex addition or removal~\citep{furini2019maximum}, remains largely unexplored.

\bibliographystyle{apalike} 
\bibliography{ORcite}

@article{zhou2022effective,
  title={An effective branch-and-bound algorithm for the maximum s-bundle problem},
  author={Zhou, Yi and Lin, Weibo and Hao, Jin-Kao and Xiao, Mingyu and Jin, Yan},
  journal={European Journal of Operational Research},
  volume={297},
  number={1},
  pages={27--39},
  year={2022},
  publisher={Elsevier}
}

@article{san2023clisat,
  title={CliSAT: A new exact algorithm for hard maximum clique problems},
  author={San Segundo, Pablo and Furini, Fabio and {\'A}lvarez, David and Pardalos, Panos M},
  journal={European Journal of Operational Research},
  volume={307},
  number={3},
  pages={1008--1025},
  year={2023},
  publisher={Elsevier}
}

@book{downey2012parameterized,
  title={Parameterized complexity},
  author={Downey, Rodney G and Fellows, Michael Ralph},
  year={2012},
  publisher={Springer Science \& Business Media}
}

@misc{zhu2025reduction,
author = {Zhu, Chenghao and Zhou, Yi and Jiang, Haoyu},
title = {A reduction-based algorithm for the clique interdiction problem},
year = {2025},
isbn = {978-1-956792-06-5},
booktitle = {Proceedings of the Thirty-Fourth International Joint Conference on Artificial Intelligence},
articleno = {1003},
numpages = {9},
location = {Montreal, Canada},
series = {IJCAI '25}
}

@inproceedings{chang2019efficient,
author = {Chang, Lijun},
title = {Efficient Maximum Clique Computation over Large Sparse Graphs},
year = {2019},
isbn = {9781450362016},
publisher = {Association for Computing Machinery},
address = {New York, NY, USA},
booktitle = {Proceedings of the 25th ACM SIGKDD International Conference on Knowledge Discovery \& Data Mining},
pages = {529–538},
numpages = {10},
location = {Anchorage, AK, USA},
series = {KDD '19}
}

@article{LI20171,
title = {On minimization of the number of branches in branch-and-bound algorithms for the maximum clique problem},
journal = {Computers \& Operations Research},
volume = {84},
pages = {1-15},
year = {2017},
issn = {0305-0548},
author = {Chu-Min Li and Hua Jiang and Felip Manyà},
}

@article{dempe2020bilevel,
  title={Bilevel optimization: theory, algorithms, applications and a bibliography},
  author={Dempe, Stephan},
  journal={Bilevel optimization: advances and next challenges},
  pages={581--672},
  year={2020},
  publisher={Springer}
}

@article{SAMUDRALA1998287,
title = {A graph-theoretic algorithm for comparative modeling of protein structure11Edited by F. Cohen},
journal = {Journal of Molecular Biology},
volume = {279},
number = {1},
pages = {287-302},
year = {1998},
issn = {0022-2836},
author = {Ram Samudrala and John Moult}
}

@article{Anand2018,
  author    = {Rajat Anand and Dipanka Tanu Sarmah and Samrat Chatterjee},
  title     = {Extracting proteins involved in disease progression using temporally connected networks},
  journal   = {BMC Systems Biology},
  year      = {2018},
  volume    = {12},
  number    = {1},
  pages     = {78},
  issn      = {1752-0509},
  month     = jul,
  language  = {english}
}

@article{SanSegundo2015,
  author    = {Pablo {San Segundo} and Jorge Artieda},
  title     = {A novel clique formulation for the visual feature matching problem},
  journal   = {Applied Intelligence},
  year      = {2015},
  volume    = {43},
  number    = {2},
  pages     = {325--342},
  month     = sep,
  issn      = {1573-7497},
}

@inproceedings{stentiford2014face,
  title={Face recognition by detection of matching cliques of points},
  author={Stentiford, Fred},
  booktitle={Image Processing: Machine Vision Applications VII},
  volume={9024},
  pages={148--158},
  year={2014},
  organization={SPIE}
}

@article{MahdaviPajouh2020,
  author     = {Mahdavi Pajouh, Foad},
  title      = {Minimum cost edge blocker clique problem},
  journal    = {Annals of Operations Research},
  year       = {2020},
  volume     = {294},
  number     = {1},
  pages      = {345--376},
  issn       = {1572-9338},
}

@article{fischetti2018dynamic,
  title={A dynamic reformulation heuristic for generalized interdiction problems},
  author={Fischetti, Matteo and Monaci, Michele and Sinnl, Markus},
  journal={European Journal of Operational Research},
  volume={267},
  number={1},
  pages={40--51},
  year={2018},
  publisher={Elsevier}
}

@article{zhong2024interdicting,
  title={On Interdicting Dense Clusters in a Network},
  author={Zhong, Haonan and Mahdavi Pajouh, Foad and Butenko, Sergiy and Prokopyev, Oleg A},
  journal={INFORMS journal on computing},
  year={2024},
  publisher={INFORMS}
}

@inproceedings{hastad1996clique,
  title={Clique is hard to approximate within n/sup 1-/spl epsiv},
  author={Hastad, Johan},
  booktitle={Proceedings of 37th Conference on Foundations of Computer Science},
  pages={627--636},
  year={1996},
  organization={IEEE}
}

@phdthesis{becker2017bilevel,
  title={Bilevel Clique Interdiction and Related Problems},
  author={Becker, Timothy},
  year={2017},
  school={Rice University}
}

@article{furini2019maximum,
  title={The maximum clique interdiction problem},
  author={Furini, Fabio and Ljubi{\'c}, Ivana and Martin, S{\'e}bastien and San Segundo, Pablo},
  journal={European Journal of Operational Research},
  volume={277},
  number={1},
  pages={112--127},
  year={2019},
  publisher={Elsevier}
}

@misc{networkrepository,
  author       = {Ryan A. Rossi and Nesreen K. Ahmed},
  title        = {The Network Data Repository with Interactive Graph Analytics and Visualization},
  howpublished = {\url{https://networkrepository.com/}},
  year         = {2015},
  note         = {Accessed: 2025-01-12}
}

@article{furini2021branch,
  title={A branch-and-cut algorithm for the edge interdiction clique problem},
  author={Furini, Fabio and Ljubi{\'c}, Ivana and San Segundo, Pablo and Zhao, Yanlu},
  journal={European Journal of Operational Research},
  volume={294},
  number={1},
  pages={54--69},
  year={2021},
  publisher={Elsevier}
}

@incollection{luo2024faster,
  title={A Faster Branching Algorithm for the Maximum k-Defective Clique Problem},
  author={Luo, Chunyu and Zhou, Yi and Wang, Zhengren and Xiao, Mingyu},
  booktitle={ECAI 2024},
  pages={4132--4139},
  year={2024},
  publisher={IOS Press}
}

@article{fischetti2017new,
  title={A new general-purpose algorithm for mixed-integer bilevel linear programs},
  author={Fischetti, Matteo and Ljubi{\'c}, Ivana and Monaci, Michele and Sinnl, Markus},
  journal={Operations Research},
  volume={65},
  number={6},
  pages={1615--1637},
  year={2017},
  publisher={INFORMS}
}

@article{MATTIA202448,
title = {Reformulations and complexity of the clique interdiction problem by graph mapping},
journal = {Discrete Applied Mathematics},
volume = {354},
pages = {48-57},
year = {2024},
note = {18th Cologne-Twente Workshop on Graphs and Combinatorial Optimization (CTW 2020)},
issn = {0166-218X},
author = {Sara Mattia}
}

@inproceedings{zhou2021improving,
  title={Improving maximum k-plex solver via second-order reduction and graph color bounding},
  author={Zhou, Yi and Hu, Shan and Xiao, Mingyu and Fu, Zhang-Hua},
  booktitle={Proceedings of the AAAI Conference on Artificial Intelligence},
  volume={35},
  number={14},
  pages={12453--12460},
  year={2021}
}

@article{chen2021computing,
  title={Computing maximum k-defective cliques in massive graphs},
  author={Chen, Xiaoyu and Zhou, Yi and Hao, Jin-Kao and Xiao, Mingyu},
  journal={Computers \& Operations Research},
  volume={127},
  pages={105131},
  year={2021},
  publisher={Elsevier}
}

@article{gschwind2018maximum,
  title={Maximum weight relaxed cliques and Russian doll search revisited},
  author={Gschwind, Timo and Irnich, Stefan and Podlinski, Isabel},
  journal={Discrete Applied Mathematics},
  volume={234},
  pages={131--138},
  year={2018},
  publisher={Elsevier}
}

@inproceedings{li2010efficient,
  title={An efficient branch-and-bound algorithm based on maxsat for the maximum clique problem.},
  author={Li, Chu Min and Quan, Zhe},
  booktitle={AAAI},
  volume={10},
  pages={128--133},
  year={2010}
}

@article{zhou2017frequency,
  title={Frequency-driven tabu search for the maximum s-plex problem},
  author={Zhou, Yi and Hao, Jin-Kao},
  journal={Computers \& Operations Research},
  volume={86},
  pages={65--78},
  year={2017},
  publisher={Elsevier}
}

@article{pattillo_clique_2013,
  title={On clique relaxation models in network analysis},
  author={Pattillo, Jeffrey and Youssef, Nataly and Butenko, Sergiy},
  journal={European Journal of Operational Research},
  volume={226},
  number={1},
  pages={9--18},
  year={2013},
  publisher={Elsevier}
}

@article{chang2023efficient,
  title={Efficient Maximum k-Defective Clique Computation with Improved Time Complexity},
  author={Chang, Lijun},
  journal={Proceedings of the ACM on Management of Data},
  volume={1},
  number={3},
  pages={1--26},
  year={2023},
  publisher={ACM New York, NY, USA}
}

@inproceedings{gao2022exact,
  title={An exact algorithm with new upper bounds for the maximum k-defective clique problem in massive sparse graphs},
  author={Gao, Jian and Xu, Zhenghang and Li, Ruizhi and Yin, Minghao},
  booktitle={Proceedings of the AAAI Conference on Artificial Intelligence},
  volume={36},
  number={9},
  pages={10174--10183},
  year={2022}
}

@article{dai2024theoretically,
  title={Theoretically and Practically Efficient Maximum Defective Clique Search},
  author={Dai, Qiangqiang and Li, Ronghua and Cui, Donghang and Wang, Guoren},
  journal={Proceedings of the ACM on Management of Data},
  volume={2},
  number={4},
  pages={1--27},
  year={2024},
  publisher={ACM New York, NY, USA}
}

@article{wu2015review,
  title={A review on algorithms for maximum clique problems},
  author={Wu, Qinghua and Hao, Jin-Kao},
  journal={European Journal of Operational Research},
  volume={242},
  number={3},
  pages={693--709},
  year={2015},
  publisher={Elsevier}
}

@article{turan1941external,
  title={On an external problem in graph theory},
  author={Tur{\'a}n, Paul},
  journal={Mat. Fiz. Lapok},
  volume={48},
  pages={436--452},
  year={1941}
}

\end{document}